\documentclass{IEEEtran}
\usepackage{graphicx}
\usepackage{epsfig} 
\usepackage{times} 
\usepackage{amsthm}
\usepackage{amsmath,amssymb,amsfonts}
\usepackage[colorlinks=false, urlcolor=blue, pdfborder={0 0 0}]{hyperref}
\usepackage{enumerate}
\usepackage{color}
\usepackage[ruled,vlined]{algorithm2e}
\usepackage{algpseudocode}
\usepackage{subfigure}
\usepackage{comment}
\usepackage{cite}
\usepackage{flushend}
\title{Opacity Enforcing Supervisory Control using Non-deterministic Supervisors}

\author{Yifan~Xie,~\IEEEmembership{Student Member,~IEEE,}
	Xiang~Yin,~\IEEEmembership{Member,~IEEE,}
    Shaoyuan~Li,~\IEEEmembership{Senior Member,~IEEE}
	\thanks{This work was  supported by the National Natural Science Foundation of China (62061136004, 62173226, 61803259) and by Shanghai Jiao Tong University Scientific and Technological Innovation Funds.
	}
	\thanks{Yifan Xie, Xiang Yin and Shaoyuan Li  are with Department of Automation
		and Key Laboratory of System Control and Information Processing,
		Shanghai Jiao Tong University, Shanghai 200240, China.
		{\tt\small  \{xyfan1234,yinxiang,syli\}@sjtu.edu.cn.}
		
		(Corresponding Author: Xiang Yin)}
}

\newtheorem{mydef}{Definition}
\newtheorem{mythm}{Theorem}
\newtheorem{myprob}{Problem}

\newtheorem{mycol}{Corollary}
\newtheorem{mypro}{Proposition}
\newtheorem{myexm}{Example}
\newtheorem{remark}{Remark}

\begin{document}
	
	\maketitle

\begin{abstract}
In this paper, we investigate the enforcement of opacity via supervisory control in the context of discrete-event systems.
A system is said to be opaque if the intruder, which is modeled as a passive observer, can never infer confidently that the system is at a secret state.
The design objective is to synthesize a supervisor such that the closed-loop system is opaque even when the control policy is publicly known.
In this paper, we propose a new approach for enforcing opacity using \emph{non-deterministic supervisors}.
A non-deterministic supervisor is a decision mechanism that provides \emph{a set of control decisions} at each instant,
and randomly picks a specific control decision from the decision set to actually control the plant.
Compared with the standard deterministic control mechanism,
such a non-deterministic control mechanism can enhance the plausible deniability of the controlled system as the online control decision  is a random realization and cannot be implicitly inferred from the control policy. We provide a sound and complete algorithm for synthesizing a non-deterministic opacity-enforcing supervisor.
Furthermore, we show that non-deterministic supervisors are strictly more powerful than deterministic supervisors in the sense that
there may exist a non-deterministic opacity-enforcing supervisor even when deterministic supervisors cannot enforce opacity.
\end{abstract}

	
\section{Introduction}
\label{sec:introduction}
\IEEEPARstart{I}{nformation} security and privacy have become increasingly important issues in the analysis and design of modern engineering systems due to potential malicious attacks and information leakages in networks.
In this paper, we investigate an important information-flow security property called \emph{opacity} in the context of Discrete-Event Systems (DES).
In this framework, a dynamic system is modeled as a DES and an intruder is modeled as a passive observer that monitors the behavior of the dynamic system via observable events.
Essentially, opacity is a confidential property capturing whether or not the system can always deny of the possibility of executing of a secret behavior even when it may be true, i.e., it holds the plausible deniability for secret behaviors. Therefore, a system is said to be opaque with respect to a set of secret states if the intruder can never know for sure that the system is visiting a secret state.

Due to the increasing demands for security certification in safety-critical systems, the notion of opacity has drawn considerable attention in the past years in the literature;
see, e.g.,  \cite{mazare2004using,badouel2007concurrent,bryans2008opacity}.
In particular, in the context of DES, different notions of opacity have been studied, including, e.g., current-state opacity \cite{lin2011opacity}, initial-state opacity \cite{saboori2013verification}, $K$-step and infinite-step opacity \cite{yin2017new}.
The verification of opacity has also been studied for different DES models including Petri nets \cite{tong2017verification,lefebvre2019exposure,saadaoui2020current}, stochastic DES \cite{keroglou2016probabilistic,chen2017quantification,wu2018privacy,yin2019infinite}, real-time  systems \cite{zhan2018} and networked DES \cite{yao2020initial,lin2020information}.
More recently, the notion of opacity has been extended to linear/nonlinear systems with infinite-states and continuous dynamics \cite{ramasubramanian2020notions,an2020opacity,yin2020approximate}.
The reader is referred to the comprehensive  surveys \cite{jacob2016overview,lafortune2018history} and the  textbook \cite{hadjicostis2020estimation} for recent advances on this active research area.

Given an open-loop system that is verified to be non-opaque, one important problem is to \emph{enforce} opacity via some enforcement mechanisms.
This is also referred to as the \emph{synthesis problem}, which  is a very active research topic in the literature and many different enforcement mechanisms have been proposed.
For example, \cite{cassez2012synthesis,zhang2015max,behinaein2019optimal}  consider the enforcement of opacity via dynamic masks that change the output information dynamically.
The idea of  changing the output information has also been leveraged by  using insertion functions \cite{ji2018enforcement,wu2018synthesis,keroglou2018insertion,liu2020k} and event shuffles \cite{barcelos2018enforcing}.
In addition, event delays is also used to enforce opacity in \cite{falcone2015enforcement}.

One of the most widely investigated opacity enforcement mechanism is via the \emph{supervisory control theory} \cite{saboori2011opacity,yin2016uniform,he2021performance}.
In this framework, a supervisor is used to restrict the behavior of the system such that the closed-loop system is opaque   \cite{saboori2011opacity,darondeau2014enforcing,zinck2020enforcing}.
For example, in \cite{takai2008formula} a formula for controllable and opaque sublanguage is provided.
In \cite{dubreil2010supervisory}, the authors solve the opacity control problem by assuming that all controllable events are observable and the observation of the intruder is included in the observation of the supervisor.
In \cite{yin2016uniform}, a uniform approach is provided to solve the opacity-enforcing control problem without the assumption that controllable events are observable; however, it assumes that the observations of the supervisor and the intruder are equivalent.
Recently in \cite{tong2018current}, the authors provide an algorithm for synthesizing an opacity enforcing supervisor without any assumption on event sets. However, it needs to assume that the control policy is not publicly known, which reduces the problem to the computation for a maximal controllable and observable sublanguage of the supremal opaque sublanguage.

Note that all existing works on opacity-enforcing supervisory control considers deterministic supervisors, which issue
a specific control decision at each instant. However, such a
deterministic decision mechanism may decrease the plausible
deniability of the system. This is because, by knowing the
control policy and by observing the occurrences of observable
events, the intruder can recover the control decision made by
the supervisor and, therefore, obtain a better state-estimate of
the system.

In this paper, we propose to use \emph{non-deterministic supervisors}, for the first time, to enforce opacity.
Unlike a deterministic supervisor that issues a specific control decision at each instant, a non-deterministic supervisor provides \emph{a set of control decisions} at each instant and
the specific control decision applied is chosen randomly via a ``coin toss" manner.
In other words, even if the intruder knows the control policy, it still does not know the specific control decision applied as it is decided randomly on-the-fly.
Compared with the deterministic control mechanism, the non-deterministic control mechanism can significantly enhance the plausible deniability of the system, and, therefore, is more likely to enforce opacity.

The main contribution of this paper is that we provide an algorithmic correct-by-construction procedure for  synthesizing a non-deterministic supervisor that enforces opacity.
This problem is fundamentally more challenging than  the deterministic case as the observation of the supervisor and the intruder are \emph{incomparable}.
Specifically, although  the specific control decision applied is unknown \emph{a priori}, the supervisor will know this online choice after it is chosen. This information, however, is not available to the intruder.
Hence, the supervisor's knowledge is strictly more than that of the intruder.
In the standard opacity-enforcing  control problem, it is sufficient   to know the state-estimate of the system, which is not sufficient in our setting due to the issue of incomparable information.
To address this issue, we propose a new information-state that not only contains the state-estimate from the supervisor's point of view, but also contains the estimate of the supervisor's estimate from the intruder's point of view.
In other words, the control decision should be made not only based on what the supervisor thinks about the plant, but also based on what the intruder thinks about the supervisor.
Based on the proposed new information-state, we provide a sound and complete approach that synthesizes a non-deterministic supervisor enforcing opacity.
In particular, we show that using non-deterministic supervisors is strictly more powerful than using deterministic supervisors, in the sense that,
there may exist a non-deterministic opacity-enforcing supervisor even when  deterministic supervisors cannot enforce opacity.

We note that the notion of non-deterministic supervisors was originally proposed in \cite{inan1994nondeterministic}   to solve the standard supervisory control problem  for safety and non-blockingness under partial observation.
This approach was extended by \cite{kumar2005polynomial}.
Non-deterministic control mechanism has also been used for (bi)similarity enforcing supervisory control  problems with non-deterministic models and specifications \cite{fabian1996non,zhou2006control,takai2019bisimilarity,takai2020synthesis,farhat2020control}.
However, to the best of our knowledge, non-deterministic supervisors  have never been applied to the opacity-enforcement problem.
More importantly, the essence of why we use non-deterministic supervisors here is to enhance the plausible deniability of the system, which is fundamentally different from the essence of the existing works.

The rest of this paper is organized as follows.
In Section~\ref{section2}, we introduce some necessary preliminaries.
In Section \ref{section3}, we first provide a motivating example to illustrate the advantage of non-deterministic supervisors. Then we formally present the non-deterministic control mechanism and  formulate the corresponding opacity enforcement control problem.
In Section \ref{section4}, we propose a new type of information state (IS) that captures both the knowledge of the supervisor and the knowledge of the intruder and analyze the underlying information-flow.
Then we restrict our attention to the class of  information-state-based supervisors and discuss how an IS-based supervisor can be encoded as or be decoded from an IS-mapping.
In Section \ref{section5},  we propose an algorithm to synthesize an IS-based non-deterministic opacity-enforcing supervisor based on the structure of the generalized bipartite transition system.
In Section \ref{section6}, we prove the correctness   of the synthesis procedure proposed in Section \ref{section5} by showing that restricting our attention to IS-based supervisors is without loss of generality.
Finally, we conclude the paper in Section \ref{section7}.
Preliminary and partial versions of some of the results in this paper are presented in \cite{xie2020ifac}.
First, all definitions, notations and theorems in \cite{xie2020ifac} have been reformulated in a more uniform manner.
More importantly, the result in \cite{xie2020ifac} is only sound as it restricts the solution space to a finite space a priori.
In this work, we show that restricting to IS-based supervisor is without loss of generality using new techniques developed based on IS-mappings.
This new result establishes both the soundness and the completeness of the synthesis algorithm, i.e., the non-deterministic synthesis problem is completely solved.

\section{Preliminaries}\label{section2}
\subsection{System Model}
Let $\Sigma$ be a finite set of events.
A string over $\Sigma$ is a finite sequence $s=\sigma_1\cdots\sigma_n,\sigma_i\in \Sigma$. We denote by $\Sigma^*$ the set of all strings over $\Sigma$ including the empty string $\epsilon$.
A language $L\subseteq \Sigma^*$ is a set of strings.
For two languages $L_1$ and $L_2$, their concatenation is $L_1L_2=\{s_1s_2\in \Sigma^*: s_1\in L_1,s_2\in L_2\}$.
The prefix-closure of language $L$ is defined by $\overline{L}=\{v\in \Sigma^*:  \exists u\in \Sigma^* \text{ s.t. } vu\in L\}$.

We assume basic knowledge of DES and use common notations; see, e.g., \cite{Lbook}.
A DES is modeled as a deterministic finite-state automaton
\[
G=(X,\Sigma,\delta,x_0),
\]
where $X$ is the finite set of states, $\Sigma$ is the finite set of events, $\delta:X\times\Sigma\rightarrow X$ is the partial transition function, where $\delta(x,\sigma)=y$ means that there is a transition labeled by event $\sigma$ from state $x$ to $y$,
and $x_0\in X$ is the initial state.
The transition function can also be extended to $\delta:X\times \Sigma^*\to X$ in the usual manner \cite{Lbook}.
For simplicity, we write $\delta(x,s)$ as $\delta(s)$ when $x=x_0$.
The language generated by $G$ is defined by $\mathcal{L}(G):=\{s\in \Sigma^*: \delta(x_0,s)!\}$, where $!$ means ``is defined''.

When the system is partially observed, $\Sigma$ is partitioned into two disjoint sets:
$\Sigma=\Sigma_o\dot{\cup}\Sigma_{uo}$, where $\Sigma_o$ is the set of observable events and $\Sigma_{uo}$ is the set of unobservable events.
The natural projection $P:\Sigma^*\rightarrow \Sigma_o^*$ is defined by
\[
P(\epsilon)=\epsilon  \text{ and }
P(s\sigma)=\left\{
\begin{aligned}
&P(s)\sigma &\text{if }& \sigma\in\Sigma_o \\
&P(s)       &\text{if }& \sigma\in\Sigma_{uo}
\end{aligned}
\right.
\]
The natural projection is also extended to $P:2^{\Sigma^*}\to 2^{\Sigma_o^*}$ by $P(L)=\{ P(s): s\in L \}$.

\subsection{Deterministic Supervisory Control}
In the framework of supervisory control, a supervisor dynamically enables/disables controllable events based on its observation.
Formally, we assume that the events set is further  partitioned as $\Sigma=\Sigma_c\dot{\cup}\Sigma_{uc}$,
where $\Sigma_c$ is the set of controllable events and $\Sigma_{uc}$ is the set of uncontrollable events.
A control decision $\gamma\in2^{\Sigma}$ is a set of events such that $\Sigma_{uc}\subseteq\gamma$, i.e.,  uncontrollable events can never be disabled.
We define $\Gamma=\{\gamma\in2^{\Sigma}:\Sigma_{uc}\subseteq\gamma\}$ as the set of  control decisions or control patterns.
A \emph{deterministic supervisor} is a function $S:P(\mathcal{L}(G))\rightarrow\Gamma$.
The language generated by the controlled system, denoted by $\mathcal{L}(S/G)$, is defined recursively by
\begin{itemize}
	\item
	$\epsilon \in \mathcal{L}(S/G)$; and
	\item
	For any $s\in \Sigma^*, \sigma\in \Sigma$, we have $s\sigma\in \mathcal{L}(S/G)$ iff
	$s\sigma\in\mathcal{L}(G), s\in \mathcal{L}(S/G)$ and $\sigma\in S(P(s))$.
\end{itemize}

\subsection{Opacity}
We assume that system $G$ has a ``secret", which is modeled as a set of secret states $X_S\subseteq X$.
Furthermore, we consider a passive \emph{intruder} having the following capabilities:
\begin{itemize}
	\item[A1]
	The intruder knows the system model;
	\item[A2]
	The intruder can observe the occurrences of observable events.
\end{itemize}
Such an intruder is essentially an outside observer or an ``eavesdropper". We say that  system $G$ is \emph{opaque} w.r.t.\ $X_S$ and $\Sigma_o$ if
\[
(\forall s\!\in\! \mathcal{L}(G): \delta(s)\!\in\! X_S)(\exists t\!\in\! \mathcal{L}(G): \delta(t)\!\notin\! X_S)[P(s)=P(t)].
\]
That is, the intruder cannot infer for sure that the system is in a secret state based on the information flow.

When the original system is not opaque, one approach is to design a supervisor $S$ such that the closed-loop system $S/G$ is opaque; this is referred to as the
\emph{opacity-enforcing control problem}.
In this setting, however, the implementation of such a supervisor may become a public information.
To capture this severe scenario, we assume:
\begin{itemize}
	\item[A3]
	The intruder knows the functionality of the supervisor, i.e., the control policy.
\end{itemize}
Note that, under the setting of deterministic supervisors, this knowledge together with the assumption that the intruder and the observer both observe $\Sigma_o$ imply that
the intruder knows precisely the control decision applied at each instant.
Therefore, to define opacity of the controlled system, we should only consider strings in $\mathcal{L}(S/G)$ rather than all strings in $\mathcal{L}(G)$.
Formally, we say that a deterministic supervisor $S:P(\mathcal{L}(G))\rightarrow\Gamma$ enforces opacity on $G$, or  the closed-loop system $S/G$ is  opaque, if
for any string $s \in  \mathcal{L}(S/G)$ such that $\delta(s) \in  X_S$,
there exists a  string $t \in  \mathcal{L}(S/G)$ such that $\delta(t) \notin  X_S$ and $P(s)=P(t)$.

Finally, we introduce some operators that will be used in this paper.
Given a set of states $m\in 2^X$, we denote by $U\!R_{\gamma}(m)$ the \emph{unobservable reach} of $m$ under control decision $\gamma \in \Gamma$, i.e.,
\begin{equation}\label{UR}
U\!R_{\gamma}(m)=\{\delta(x,w)\in X: x\in m, w\in(\Sigma_{uo}\cap\gamma)^*\}.
\end{equation}
We also denote by $N\!X_{\sigma}(m)$ the \emph{observable reach} of $m$ upon the occurrence of an observable event $\sigma\in \Sigma_o$, i.e.,
\begin{equation}\label{NX}
N\!X_\sigma(m)=\{\delta(x,\sigma)\in X:  x\in m\}.
\end{equation}

\section{Enforcing Opacity using Non-deterministic Supervisors}\label{section3}
In this section, we propose to use non-deterministic supervisors to enforce opacity.
First, we illustrate the advantage of using non-deterministic supervisors by a motivating example.
Then we formally define the functionality of the non-deterministic supervisor and opacity of non-deterministic control systems.
We formulate the corresponding opacity-enforcing supervisory control problem that we want to solve in this paper.

\subsection{Motivating Example}

\begin{figure}[!t]
	\centerline{\includegraphics[width=6.5cm]{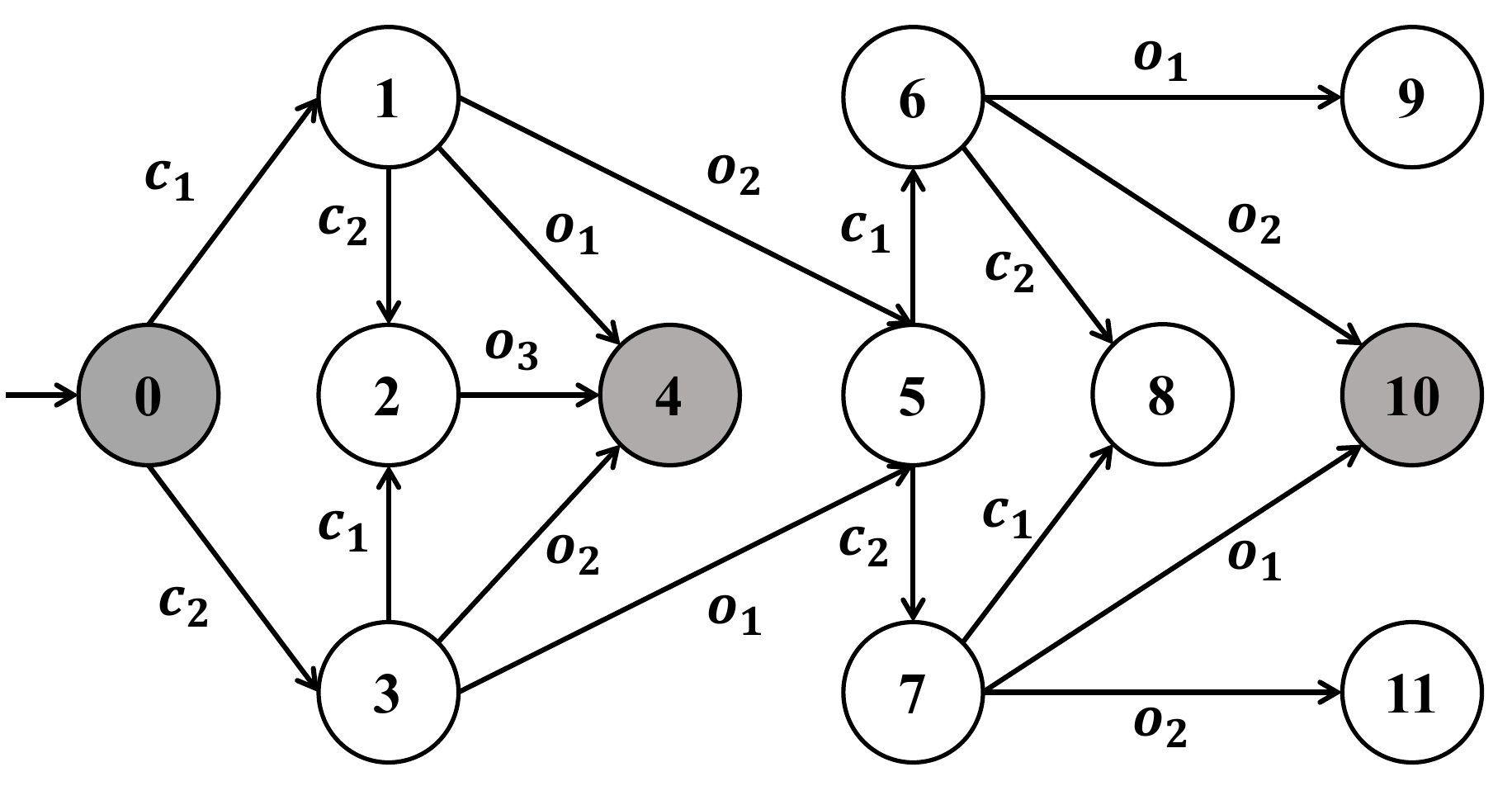}}
	\caption{System $G$ with $\Sigma_c=\{c_1,c_2\}$, $\Sigma_o=\{o_1,o_2,o_3\}$ and $X_S=\{0,4,10\}$.}\label{system}
\end{figure}

\begin{myexm}\label{exm1}\upshape
Let us consider system $G$ shown in Fig.~\ref{system}  with $\Sigma_o=\Sigma_{uc}=\{o_1,o_2,o_3\}$ and $X_S=\{0,4,10\}$.
String  $c_2o_1c_2o_1$ leads to secret state $10$ and its observation is $P(c_2o_1c_2o_1)=o_1o_1$.
By observing $o_1o_1$, the intruder cannot infer for sure that the system is in state $10$ since $P(c_2o_1c_1o_1)=o_1o_1$ and $\delta(c_2o_1c_1o_1)=9\not\in X_S$.	
However, by observing $o_3$, the intruder   knows for sure that the system is at secret state $4$ since for any string $s$ such that $P(s)=o_3$, we have $\delta(s)=4\in X_S$.
Therefore, the system is not opaque and  we need to synthesize a supervisor to protect the system from  revealing secret $4$.
	
For this system, however, we cannot even synthesize a deterministic supervisor to enforce opacity.
To see this clearly, let us evaluate what the supervisor can do initially.
	We have $\Gamma=\{ \emptyset,  \{c_1\},  \{c_2\}, \{c_1,c_2\}   \}$.\footnote{For the sake of simplicity, uncontrollable events are omitted in each control decision, i.e., $\emptyset$ standards for $\{o_1,o_2,o_3\}$ in this example.}
	Clearly, the supervisor cannot choose $\emptyset$ as the initial control decision; otherwise secret state $0$ will be the only reachable state.
	Also, the supervisor cannot make $\{c_1\}$ initially.
	This is because, under this control decision and by observing event $o_1$, the intruder knows for sure that the system is at state $4$ which is reached via $0\xrightarrow{c_1}1\xrightarrow{o_1}4$.
	Note that transitions $0\xrightarrow{c_2}3\xrightarrow{o_1}5$ cannot provide the plausible  deniability since $c_2$ is disabled initially.
	For the same reason, making $ \{c_2\}$ initially will also reveal the secret.
	Finally, decision $\{c_1,c_2\}$ is also problematic initially as it makes state $2$ reachable from which transition $2\xrightarrow{o_3}4$ will also reveal the secret.
	Therefore, we cannot enforce opacity for this system using a deterministic supervisor.
	
	However, one can enforce opacity using the following control mechanism.
	Initially, the supervisor randomly chooses to either  enable $c_1$ or enable $c_2$, but not enable both simultaneously.
	In other words, the control policy initially is a set $\{ \{c_1\},\{c_2\}  \}$ and the specific choice is made randomly on-the-fly.
	Therefore, upon the occurrence of $o_1$ or $o_2$, the intruder does not know  whether this event is from state $1$ or from state $3$
	since it does not know whether or not the initial online control decision is $\{c_1\}$ or $\{c_2\}$ by just knowing the control policy $\{ \{c_1\},\{c_2\}  \}$.
	On the other hand, since $c_1$ and $c_2$ will not be enabled simultaneously,  state $2$ is not reachable; hence, event $o_3$, which reveals the secret, will also not occur.
	Then after observing $o_1$ or $o_2$, the supervisor can make decision $\{c_1,c_2\}$ deterministically, which prevents the system from revealing secret state $10$. 	
\end{myexm}

The above example shows that using a non-deterministic control mechanism is  more powerful than the deterministic one for the purpose of enforcing opacity.
This result is intuitive as opacity is essentially a confidential property.
Using a non-deterministic decision framework will, on the one hand, enhance the plausible deniability of the secret behavior of the system,
and, on the other hand, decrease the confidentiality of the intruder's knowledge about the system.
Hence, the system is more likely to be opaque under the non-deterministic   mechanism.

\subsection{Non-deterministic Supervisor}
Now, we formally define the non-deterministic supervisor and the corresponding opacity enforcement problem.

Compared with a deterministic supervisor that issues a specific control decision at each instant, a non-deterministic supervisor works as follows.
At each instant, the non-deterministic supervisor provides   \emph{a set of possible control decisions}.
Then it non-deterministically  picks a specific control decision from this set  in a ``coin-toss" manner and keeps this specific control decision until a new observable event occurs.
In other words, the control policy only determines a set of allowed decisions,
but the specific control decision chosen is unknown \emph{a priori}, which is a \emph{random realization} under the control policy.
Therefore, the supervisor makes decision not only based on observable events, but also depends on the specific control decisions chosen along the trajectory.

To define the ``history" of the supervisor, we introduce the notion of the \emph{extended string} which is an alternating sequence of control decisions and events either ending up with a control decision in the form of
\begin{equation}\label{eq:rho-1}
\rho=\gamma_0  \sigma_1\gamma_1 \cdots \sigma_n\gamma_n\in \Gamma(\Sigma\Gamma)^*
\end{equation}
or ending up with an event in the form of
\begin{equation}\label{eq:rho-2}
\rho=\gamma_0  \sigma_1\gamma_1 \cdots \sigma_n \in (\Gamma\Sigma)^*
\end{equation}
Then the set of all extended strings is
$\Gamma(\Sigma\Gamma)^*\cup  (\Gamma\Sigma)^*=(\Gamma\Sigma)^*(\Gamma\cup \{\epsilon\}) $.
For any extended string
$\rho \in (\Gamma\Sigma)^*(\Gamma\cup \{\epsilon\})$,
we denote by $\rho|_{\Sigma}$ the projection to $\Sigma^*$, i.e., $\rho|_{\Sigma}=\sigma_1\dots\sigma_n$.

Since some events are unobservable for the supervisor and the supervisor cannot update its decision upon the occurrence of an unobservable event,
similar to the natural projection, we define a new projection mapping
\begin{equation}
\mathcal{O}:
(\Gamma\Sigma)^*(\Gamma\cup \{\epsilon\})
\rightarrow
(\Gamma\Sigma_o)^*(\Gamma\cup \{\epsilon\})
\end{equation}
such that, for any extended string, projection  $\mathcal{O}$ erases each unobservable event together with its successor control decision (if there exists one).
Formally, for extended string $\rho$ in the form of Equation~\eqref{eq:rho-1},
let $1\leq i_1< i_2<\cdots< i_k \leq n$ be all indices such that $\sigma_{i_j}\in \Sigma_o$.
Then  we have
\begin{equation}
\mathcal{O}(\rho)=
\gamma_0(\sigma_{i_1}\gamma_{i_1})(\sigma_{i_2}\gamma_{i_2})\cdots(\sigma_{i_k}\gamma_{i_k})
\end{equation}
and for extended string $\rho$ in the form of Equation~\eqref{eq:rho-2}, we have
\begin{equation}
\mathcal{O}(\rho)=
\gamma_0(\sigma_{i_1}\gamma_{i_1})(\sigma_{i_2}\gamma_{i_2})\cdots(\sigma_{i_{k-1}}\gamma_{i_{k-1}})\sigma_{i_k}.
\end{equation}

From the supervisor's point of view, each decision is made immediately (by first providing a set of decisions and then picking one from the set) after observing an observable event. Therefore, the supervisor should make decision based on alternating sequences that end  up with observable events.
Hence, the non-deterministic supervisor is defined as a function
\begin{equation}\label{eq:nb-def}
S_N: (\Gamma \Sigma_o)^* \to 2^{\Gamma}
\end{equation}
that maps an observable extended string $\mathcal{O}(\rho)=\gamma_0  \sigma_1\gamma_1 \cdots \sigma_n  \in(\Gamma\Sigma_o)^*$, which is referred to a \emph{decision history}, to a set of possible control decisions.
This definition essentially says that, although the control policy is non-deterministic,
the supervisor knows the \emph{realization}, i.e., which specific decision was picked at each previous instant. This is a reasonable setting as the supervisor always knows what it actually picks.
Now, we define the language generated by a non-deterministic supervisor.

\begin{mydef}\label{extended string}\upshape
	Let $S_N$ be a non-deterministic supervisor.
	The set of extended strings generated by the closed-loop system,  denoted by $\mathcal{L}_e(S_N/G)$, is defined recursively by:
	\begin{itemize}
		\item
		$\epsilon\in\mathcal{L}_e(S_N/G)$;
		\item
		$\gamma_0\in\mathcal{L}_e(S_N/G)$ if $\gamma_0\in S_N(\epsilon)$;
		\item
		For any $\rho =\gamma_0\sigma_1\gamma_1\cdots\sigma_n\gamma_n \sigma_{n+1} \in  (\Gamma\Sigma)^*$,
		we have $\rho\in \mathcal{L}_e(S_N/G)$, if and only if
		\begin{itemize}
			\item
			$ \gamma_0\sigma_1\gamma_1 \dots \sigma_n \gamma_n    \in \mathcal{L}_e(S_N/G)$; and	
			\item
			$\sigma_1\cdots\sigma_n\sigma_{n+1}\in \mathcal{L}(G)$; 	
			\item
			$\sigma_{n+1}\in \gamma_n$.
		\end{itemize}
		\item
	For any $\rho =\gamma_0\sigma_1\gamma_1\cdots\sigma_n\gamma_n \sigma_{n+1}\gamma_{n+1} \in  \Gamma(\Sigma\Gamma)^*$,
	we have $\rho\in \mathcal{L}_e(S_N/G)$, if and only if
	\begin{itemize}
		\item
		$ \gamma_0\sigma_1\gamma_1\cdots\sigma_n\gamma_n \sigma_{n+1}   \in \mathcal{L}_e(S_N/G)$; and	
		\item
		$
		\gamma_{n+1} \!\!\in\!\! \left\{\!
		\begin{aligned}
		&\!\{ \gamma_n  \} &\!\!\!\!\text{if }\!&  \sigma_{n+1} \!\in\!\Sigma_{uo} \\
		& \!S_N(\!\mathcal{O}(\gamma_0\sigma_1\gamma_1 \dots \sigma_n \gamma_n   \sigma_{n+1}\! )\!)
		&\!\!\!\!\text{if }\!&  \sigma_{n+1} \!\in\!\Sigma_{o}
		\end{aligned}
		\right.
		$
		\end{itemize}
	\end{itemize}
Then
a string $s\in \Sigma^*$ is said to be \emph{generated} by $S_N/G$ if there exists an extended string $\rho \in \mathcal{L}_e(S_N/G)$ such that $\rho|_\Sigma=s$.
We define
$\mathcal{L}(S_N/G)=\{\rho|_{\Sigma} \in \Sigma^*: \rho\in \mathcal{L}_e(S_N/G)  \}$ as the language generated by the closed-loop system.
\end{mydef}

The intuition of the above definition is as follows.
Initially, the first control decision should be included in the initial set of control decisions provided by $S_N$.
When extended string $ \gamma_0\sigma_1\gamma_1 \dots \sigma_n \gamma_n$ is executed, the next event $\sigma_{n+1}$ should be both feasible in the plant and enabled by the control decision applied currently, i.e.,  $\gamma_n$.
Furthermore, if $\sigma_{n+1}$ is unobservable, then the supervisor should not change the control decision, i.e., $\gamma_{n+1}=\gamma_n$.
On the other hand, if $\sigma_{n+1}$ is observable, then the supervisor may choose a specific control decision from
the new set of all possible control decisions provided by $S_N$, i.e.,
$\gamma_{n+1}\in S_N( \mathcal{O}(\gamma_0\sigma_1\gamma_1 \dots \sigma_n \gamma_n\sigma_{n+1})   )$.
We denote by $\mathcal{L}_e^o(S_N/G)$ the set of extended strings that end up with observable events including the empty string, i.e.,
\[
\mathcal{L}_e^o(S_N/G)=\mathcal{L}_e(S_N/G)\cap (\{\epsilon\}\cup (\Gamma\Sigma)^*(\Gamma\Sigma_o) ).
\]
We also denote by   $ \mathcal{L}_e^d(S_N/G )$ the set of extended strings that end up with control decisions, i.e.,
\[
\mathcal{L}_e^d(S_N/G)=\mathcal{L}_e(S_N/G)\cap \Gamma(\Sigma\Gamma)^*.
\]
The supervisor always issues a decision (first generates a set of control decisions and then randomly picks one)  when an extended string $\rho$ in $\mathcal{L}_e^o(S_N/G)$ is generated. Then
for any observable extended string $\rho\in  \mathcal{O}(\mathcal{L}_e^o(S_N/G))$, we define
\begin{align}
\hat{\mathcal{E}}_S(\rho)
\!=\! \left\{
\delta( \rho'|_{\Sigma}  )  \!\in\! X  :
\exists   \rho'\in \mathcal{L}_{e}^o (S_N/G)
\text{ s.t. }\mathcal{O}(\rho')=\rho
\right\}
\end{align}
 as the set of all possible states that can be reached immediately after observing the last event from the supervisor's point of view,
i.e., the state estimate of the supervisor without the unobservable tail.

Once the supervisor issues the last control decision, the set of states that can be reached unobservably can be determined.
Formally,  for any  extended string $\rho\in  \mathcal{O}(\mathcal{L}_e^d(S_N/G))$, we define
\begin{align}
\mathcal{E}_S(\rho)
 =  \{  \delta( \rho'|_{\Sigma}  ) \!\in\! X   : \exists   \rho'\in \mathcal{L}_e(S_N/G) \text{ s.t. }\mathcal{O}(\rho')=\rho        \}
\end{align}
as the state-estimate of the supervisor with the unobservable tail included.
These two state estimates can be computed recursively as follows \cite{yin2016uniform}:
\begin{itemize}
	\item
	$\hat{\mathcal{E}}_S(\epsilon)=\{x_0\}$;
	\item
	$\mathcal{E}_S(\rho')= U\!R_{\gamma}( \hat{\mathcal{E}}_S(\rho)) $
	for $\rho'=\rho\gamma\in\mathcal{O}(\mathcal{L}_e^d(S_N/G))$;
	\item
	$\hat{\mathcal{E}}_S(\rho'')=  N\!X_{\sigma}( \mathcal{E}_S(\rho')  )$
	for $\rho''=\rho'\sigma\in  \mathcal{O}(\mathcal{L}_e^o(S_N/G))$.
\end{itemize}
Here we use subscript ``$S$" to emphasize that these state-estimates are from the supervisor's point of view.

\begin{myexm}\upshape
Still consider system $G$ in Fig.~\ref{system} with $\Sigma_c=\{c_1,c_2\}$ and $\Sigma_o=\{o_1,o_2,o_3\}$.
Suppose that the initial non-deterministic decision set is
$S_N(\epsilon)=\{ \gamma_1,\gamma_2,\gamma_3,\gamma_4   \}$, where $\gamma_1=\emptyset, \gamma_2=\{c_1\}, \gamma_3=\{c_2\}$ and $\gamma_4=\{c_1,c_2\}$.
Then we have $\gamma_1,\gamma_2,\gamma_3,\gamma_4\in \mathcal{L}_e(S_N/G)$.
Suppose that the supervisor chooses $\gamma_2$ initially.
Then we have  $\gamma_2 c_1 \in  \mathcal{L}_e(S_N/G)$ and since $c_1$ is unobservable, we have $\gamma_2 c_1\gamma_2\in \mathcal{L}_e(S_N/G)$. When $o_2$ occurs,
$\rho=\gamma_2 c_1 \gamma_2 o_2 \in \mathcal{L}_e^o(S_N/G)$ becomes the first extended string that ends up with an observable event.
Then the information available to the supervisor is  $\mathcal{O}(\rho)=\gamma_2 o_2$.
The state estimate of the supervisor is
$\hat{\mathcal{E}}_S(\gamma_2 o_2)=N\!X_{o_2}(U\!R_{\gamma_2}(\hat{\mathcal{E}}(\epsilon)))=N\!X_{o_2}(\{0,1\}  )=\{5\}$, i.e.,
the supervisor knows for sure that system is at state $5$ by first choosing $\gamma_2$ and then observing $o_2$.

Suppose that the supervisor then issues $\gamma_4$ deterministically, i.e., $S_N(\gamma_2 o_2)=\{\gamma_4\}$ and the supervisor can only choose $\gamma_4$;
this yields extended string $\rho'= \gamma_2 c_1 \gamma_2 o_2\gamma_4 \in \mathcal{L}_e^d(S_N/G)$
With the last control decision information attached, the information available to the supervisor is
$\mathcal{O}(\rho')=\gamma_2 o_2 \gamma_4$.
Then the state estimate of the supervisor is $\mathcal{E}_S(\gamma_2 o_2 \gamma_4)=U\!R_{\gamma_4}(\hat{\mathcal{E}}_S(\gamma_2 o_2))=\{5,6,7,8\}$.

Again,   extended string $\rho''=  \gamma_2 c_1 \gamma_2 o_2\gamma_4c_1 \gamma_4 o_1 \in \mathcal{L}_e^o(S_N/G)$ can be generated with $\mathcal{O}(\rho'')=\gamma_2 o_2 \gamma_4 o_1$. Then the state estimate of the supervisor becomes $\hat{\mathcal{E}}_S(\gamma_2 o_2 \gamma_4 o_1)=N\!X_{o_1}(\mathcal{E}_S(\gamma_2 o_2 \gamma_4))=\{9,10\}$.
\end{myexm}

\begin{remark}\upshape
Finally, we note that some non-deterministic control decision sets may contain redundancy, i.e., for $\{\gamma_1,\dots,\gamma_n\} \in 2^{\Gamma}$,
$\gamma_i\subset \gamma_j$  for some $i,j=1,\dots,n$.
In this case, removing $\gamma_i$ from the  non-deterministic control decision set does not change the behavior of the closed-loop system.
Formally, we say that a non-deterministic control decision set  $\{\gamma_1,\dots,\gamma_n\} \in 2^{\Gamma}$ is \emph{irredundant} if its elements are incomparable, i.e.,  $\forall i,j=1,\dots,n: \gamma_i  \not\subset  \gamma_j$.
For the sake of simplicity and without loss of generality, hereafter, we   only consider irredundant non-deterministic control decision sets.
\end{remark}

\begin{remark}\upshape
Note that our definition of non-deterministic supervisor in Equation~\eqref{eq:nb-def} is language-based, which may require infinite memory to realize.
However, we will show later in the paper that finite-memory supervisors are always sufficient for our purpose. For this case, one may also realize a non-deterministic supervisor by a non-deterministic finite-state automaton and the closed-loop behavior can be then computed by taking the synchronous composition between the plant and the supervisor automaton.
\end{remark}
\subsection{Opacity of Non-deterministic Control Systems}
In the definition of opacity for the standard deterministic setting, the intruder model has been specified by A1-A3.
Here, we still consider the same intruder model, but we explain A3 more clearly in the non-deterministic setting:
\begin{itemize}
	\item[A3$'$]
	The intruder knows the functionality of the supervisor.
	That is, the intruder knows   the set of all possible control decisions the supervisor may pick according to the control policy.
	However, it does not know which specific control decision the supervisor picks online.
\end{itemize}
This assumption is reasonable in many applications as long as the communication channel between supervisor and the actuator is reliable.
Then under this setting, when the supervisor observes $\rho\in \mathcal{O}(\mathcal{L}_e(S_N/G))$, the intruder can only observes
$\rho|_{\Sigma} \in P(\mathcal{L}(S_N/G))$.
Therefore, the state estimate of the intruder essentially is more uncertain, which needs to estimate all possible  realizations consistent with the  control policy and the observation.
Formally, for any observable  string $s \in P(\mathcal{L}(S_N/G))$, we define ${X}_{I}(s)$ as the state estimate of the intruder, i.e.,
\begin{equation}
{X}_{I}(s)=\{  \delta( s'  ) \!\in\! X   : \exists   s'\in \mathcal{L}(S_N/G) \text{ s.t. }   P(s')=s        \}.
\end{equation}
Then opacity of control systems under non-deterministic supervisors is defined as follows.
\begin{mydef}\label{Definition2}\upshape
	Let $S_N: (\Gamma \Sigma_o)^* \to 2^{\Gamma}$ be a non-deterministic supervisor.
	We say the closed-loop system $S_N/G$ is  opaque (w.r.t.\ $\Sigma_o$ and $X_S$) if
	$\forall  s \in P(\mathcal{L}(S_N/G))\text:{X}_{I}(s)\not\subseteq X_S$.
\end{mydef}

The state estimate of the supervisor and the state estimate of the intruder can be related as follows.
Since the intruder observes strictly less than the supervisor, its estimate of the system is essentially the union of its estimate of all possible supervisor's knowledge about the system.
To see this more clearly, for any observable  string $s \in P(\mathcal{L}(S_N/G))$, we also define
\begin{align}\label{Es}
\hat{\mathcal{E}}_I(s)
&\!=\!\{ \hat{\mathcal{E}}_S ( \rho )\in 2^X \!:\!   \rho\in \mathcal{O}(\mathcal{L}_e^o(S_N/G))\! \text{ s.t. }  \! \rho|_{\Sigma}=s   \} \\
\mathcal{E}_I(s)
&\!=\!\{ \mathcal{E}_S( \rho )\in 2^X \!:\! \rho\in \mathcal{O}(\mathcal{L}_e^d(S_N/G)) \!\text{ s.t. }\! \rho|_{\Sigma}=s   \}
\end{align}
as the intruder's estimates of the state-estimations of the supervisor.
Note that $\hat{\mathcal{E}}_I(s)$ and $\mathcal{E}_I(s)$ are respectively the state estimate immediately after observing an observable event and
the state-estimate with the unobservable tail included. Note that we use subscript ``$I$" to emphasize that these estimates are from the intruder's point of view.
Then we have the following result that connects $\mathcal{E}_I$ and $X_I$.

\begin{mypro}\label{pro1}
	For any  $s \in P(\mathcal{L}(S_N/G))$,   we have
	\[
	{X}_{I}(s)= \bigcup \mathcal{E}_I(s).
	\]
\end{mypro}

\begin{proof}
By the definitions of $\mathcal{E}_I(s),\mathcal{E}_S(\rho)$, $\mathcal{L}(S_N/G)$ and mapping $\mathcal{O}$, we have
\begin{align}
 & \bigcup\mathcal{E}_I(s)\nonumber \\
=& \bigcup\{ \mathcal{E}_S( \rho )\in 2^X :   \rho\in \mathcal{O}(\mathcal{L}_e^d(S_N/G)) \text{ s.t. }   \rho|_{\Sigma}=s   \}\nonumber \\
=&\{x\in\mathcal{E}_S(\rho):\rho\in\mathcal{O}(\mathcal{L}_e^d(S_N/G))\text{ s.t. }\rho|_\Sigma=s\}\nonumber\\
=&\{
\delta(\rho'|_\Sigma):
\rho'\in\mathcal{L}_e(S_N/G)  \text{ s.t. } \mathcal{O}(\rho')|_\Sigma=s
\}\nonumber\\
=&\{\delta(s'):  s'\in\mathcal{L}(S_N/G) \text{ s.t. }P(s')=s\}\nonumber\\
=&X_I(s)\nonumber
\end{align}
This completes the proof.
\end{proof}

Given a non-opaque system, our goal is to synthesize  a non-deterministic supervisor that restricts the system behavior such that opacity is satisfied for the closed-loop system.
The opacity enforcement synthesis problem is formulated as follows.
\begin{myprob}\label{prob1}
(Opacity Enforcement Problem)
Given system $G$ and secret states $X_S\subseteq X$, synthesize a partial observation non-deterministic supervisor $S_N: (\Gamma\Sigma_o)^*\rightarrow 2^\Gamma$, such that $S_N/G$ is opaque w.r.t.\ $X_S$ and $\Sigma_o$.
\end{myprob}

\begin{remark}\upshape
Compared with deterministic supervisors, the additional power of non-deterministic supervisor, in terms of opacity enforcement, relies on assumption A3$'$. That is, the intruder is aware of the functionality of the non-deterministic supervisor but cannot eavesdrop the specific control decisions issued by the supervisor.
Note that, if the intruder is completely not aware of the functionality of the supervisor (no matter deterministic or non-deterministic), then it has to make state estimation  based on the \emph{open-loop} system $G$.
For this case, using non-deterministic supervisors does not provide any additional power compared with deterministic supervisors, and  it suffices to solve the  \emph{supervisor-unaware} deterministic opacity-enforcement problem; see, e.g., \cite{takai2008formula,tong2018current}.
If the intruder is aware of the functionality of the non-deterministic supervisor, but at the same time,  is also capable of eavesdropping the control decisions issued by the non-deterministic supervisor,
then this essentially means that that  the non-deterministic control information can be resolved by the intruder.
For this case, using non-deterministic supervisors is still the same as using deterministic supervisors in terms of the capability of enforcing opacity.
Then it suffices to solve a \emph{supervisor-aware} deterministic opacity-enforcement problem; see, e.g., \cite{dubreil2010supervisory,saboori2011opacity,yin2016uniform}.
\end{remark}

\section{Information State and its Flow}\label{section4}
In the formulation of the opacity enforcement problem, the domain of the supervisor is defined over languages.
Therefore, the solution space is infinite in general and there is no prior knowledge to bound the memory of the supervisor.
To effectively solve the synthesis  problem, in this section, we  restrict our attention to a class of \emph{information-state-based supervisors}, where the space of  information states is finite.
We first define the information state in the non-deterministic control problem and then discuss how  the selected information state evolves.
Also, we define a IS-mapping that can encode an IS-based supervisor.
Our method for synthesizing a non-deterministic supervisor is to first synthesize a IS-mapping and then encode a supervisor from it.
To this end, we finally put forward an algorithm that decodes a non-deterministic supervisor from IS-mapping.
We will show later in Section \ref{section6} that restricting our attention to IS-based  supervisors is without loss of generality for the solvability of the  general  non-deterministic supervisor opacity enforcement problem.

\subsection{Proposed Information Structure}\label{section4.1}
In the deterministic control problem, it is known that $2^X$ is sufficient to realize an opacity-enforcing supervisor \cite{yin2016uniform}.
That is, a deterministic  supervisor can be encoded as a state-based mapping
$S: 2^X \to \Gamma$,
 which can be decoded by recursively estimating the state of the system and making decision based on the  state-estimate (information-state).

In the non-deterministic control problem, the supervisor and the intruder observe different information.
Hence, the supervisor needs to make decision based on both the state estimates of itself and that of the intruder.
To separate the observation of the supervisor and the intruder, we propose  the following information-state space
\[
I:=2^X\times 2^{2^X}.
\]
Each information state $\imath\in I$ is in the form of $\imath=( m, \mathbf{m})$.
Intuitively,
the first component aims to  represent the state estimate of supervisor,
while the second component aims to  represent intruder's knowledge of the supervisor.

Formally, given a non-deterministic supervisor $S_N$ and let $\rho\in \mathcal{O}(\mathcal{L}_e^o(S_N/G))$ be a decision history observed by the supervisor.
We define
\[
\mathcal{I}_{S_N}(\rho)=( \hat{\mathcal{E}}_S(\rho) , \hat{\mathcal{E}}_I(\rho|_{\Sigma})  ) \in 2^X\times 2^{2^X}
\]
as the information-state reached by $\rho$ under $S_N$.
Clearly, we have $ \hat{\mathcal{E}}_S(\rho) \in \hat{\mathcal{E}}_I(\rho|_{\Sigma}) $ for any $\rho$ by definition.
We also define
\[
\mathcal{I}_{S_N}:=\{ \mathcal{I}_{S_N}(\rho):\rho\in \mathcal{O}(\mathcal{L}_e^o(S_N/G)) \}
\]
the set of all information-states reached by $S_N$.

\begin{mydef}\upshape
	A non-deterministic supervisor $S_N: (\Gamma\Sigma_o)^*\to 2^{\Gamma}$ is said to be \emph{information-state-based} (IS-based) if
	\begin{equation}\label{eq:IS-condition}
	\forall \rho,\rho' \!\in\! \mathcal{O}(\mathcal{L}_e^o(S_N/G)):\mathcal{I}_{S_N}(\rho)\!=\! \mathcal{I}_{S_N}(\rho')\!\Rightarrow\! S_N(\rho)\!=\!S_N(\rho').
	\end{equation}
\end{mydef}

An IS-based supervisor only makes decisions based on its current information-state rather than the entire history.
Therefore, we can encode an IS-based supervisor  as a partial IS-mapping.

\begin{mydef}\upshape
We say a partial IS-mapping $\Theta :I \rightarrow 2^\Gamma$ \emph{encodes}  supervisor $S_N: (\Gamma\Sigma_o)^*\to 2^{\Gamma}$ if
\[
\forall \rho \in  \mathcal{O}(\mathcal{L}_e^o(S_N/G)):   \Theta(    \mathcal{I}_{S_N}(\rho)        ) =   S_N(\rho).
\]
\end{mydef}

Our general approach for synthesizing a non-deterministic supervisor is to synthesize its  IS-mapping encoding.
Clearly, given an IS-based supervisor $S_N$,  we can easily encode it as an IS-mapping  $\Theta:I \rightarrow 2^\Gamma$, which is defined at each state in $\mathcal{I}_{S_N}$.
On the other hand, however, given a partial IS-mapping $\Theta:I  \rightarrow 2^\Gamma$, it is not straightforward how to \emph{decode} an IS-based supervisor from it. In fact, not every partial IS-mapping $\Theta:I  \rightarrow 2^\Gamma$ actually encodes  an IS-based supervisor.
As a necessary requirement, the partial IS-mapping should be defined at state $\imath_0=( \{x_0\},\{\{x_0\}\}  )$, which is the initial information-state of any IS-based supervisor.
Then one can argue inductively that, for any reachable information-state, the partial IS-mapping should be defined,
which suggests that the domain of the partial IS-mapping should contain the ``reachability closure" from the initial-state $\imath_0$; otherwise, the decoded supervisor   will ``get stuck" at those states  where the IS-mapping is undefined.

To compute such an ``reachability closure", we need to investigate how the information-state evolves.
As we discussed earlier, the first component of the information-state can be computed recursively based on $\rho$.
However, the question is how to compute the second component.
To this end, we should not only know the control decision for history $\rho$, but should also know the control decisions for those $\rho'$ such that $\rho|_\Sigma=\rho'|_\Sigma$.
In the remaining part of this section, we will elaborate on how  $\hat{\mathcal{E}}_I(\rho|_{\Sigma})$ can be computed  recursively and by what information.

\subsection{Micro/Macro States and Decisions}
Before we proceed further, we define some necessary concepts.
First, we introduce the notion of micro-state, which is used to represent the knowledge of supervisor.
\begin{mydef}(Micro-State)\upshape
	A \emph{micro-state} $m \in 2^X$ is a set of states and we define $M=2^X$ as the set of micro-states.
	An \emph{augmented micro-state} $m^+=(m,\gamma) \in  2^X\times \Gamma$ is a micro-state augmented with a control decision and we define $M^+=2^X\times \Gamma$ as the set of augmented micro-states.
\end{mydef}

Then, we define the notion of macro-state, which is used to represent the knowledge of intruder about the supervisor.

\begin{mydef}(Macro-State)\upshape
	A \emph{macro-state}  $\mathbf{m}=\{  m_1 ,  m_2,$ $  \dots, m_n \}\subseteq 2^X$ is a set of micro-states and we define $\mathbb{M}=2^{2^X}$ as the set of macro-states.
	An \emph{augmented macro-state}  $\mathbf{m}^+=\{  (m_1,\gamma_1) ,  (m_2,\gamma_2) ,\dots, $ $(m_n,\gamma_n) \}\subseteq 2^X\times \Gamma$ is a set of augmented micro-states and we define $\mathbb{M}^+=2^{2^X\times \Gamma}$ as the set of augmented macro-states.
\end{mydef}

In order to estimate the knowledge of the intruder, we should not only know the decision of the supervisor at a specific micro-state,
but also should know the decisions at other micro-states   in the same macro-state, which means that these micro-states are indistinguishable from the intruder's point of view.
This leads to the notion of macro-control-decision.

\begin{mydef}(Macro-Control-Decision) \upshape
	A \emph{macro-control-decision} is a set in the form of
	\[
	d=\{(m_1,\Gamma_1),(m_2,\Gamma_2),\dots, (m_n,\Gamma_n)  \}  \subseteq 2^X\times 2^\Gamma,
	\]
	where each $(m_i,\Gamma_i)$  is a pair of micro-state and a non-deterministic control decision (a set of control decisions).
	We denote by $\!D\!=\!2^{2^X\!\times\! 2^\Gamma}\!\!$ the set of  macro-control-decisions.
\end{mydef}

Let  $\mathbf{m}=\{  m_1 ,  m_2, ,\dots, m_n \}\in \mathbb{M}$ be a macro-state and $d \in D$ be a macro-control-decision.
We say that $d$ is \emph{compatible} with $\mathbf{m}$ if it is in the form of
\[
d=\{(m_1,\Gamma_1),(m_2,\Gamma_2),\dots, (m_n,\Gamma_n)  \}  \subseteq 2^X\times 2^\Gamma,
\]
i.e., $d$ essentially assigns each micro-state   $m_i\in \mathbf{m}$ a non-deterministic control decision $\Gamma_i\in 2^{\Gamma}$.

The unobservable reach of a macro-control-decision $d\in D$ is defined by
\[
\odot(d)=\{(m',\gamma):  \exists (m,\Gamma)\in d,\gamma\in\Gamma\text{ s.t. }m'=  {U\!R}_\gamma(m)\}.
\]

Let $\mathbf{m}^+$ be an augmented macro-state and $\sigma\in\Sigma_o$ be an observable event.
Then the observable reach of $\mathbf{m}^+$ upon the occurrence of $\sigma$ is defined as
\[
\widehat{N\!X}_{\sigma}(\mathbf{m}^+ )=\{  m':  \exists (m,\gamma)\!\in\! \mathbf{m}^+ \text{ s.t. } m'\!=\!{N\!X}_\sigma(m) \wedge \sigma\!\in\! \gamma  \}.
\]

\subsection{Information-Flow Analysis}
Now, suppose that an IS-mapping $\Theta: I\to 2^{\Gamma}$ that encodes an IS-based supervisor $S_N$ is given.
Let $\textbf{m}=\{m_1,\dots, m_k\}$ be a macro-state representing the intruder's estimate of the supervisor's knowledge.
We define
\[
d_{\Theta}(\textbf{m})  =  \{ (m_1, \Theta(m_1,\mathbf{m}))  ,\dots,   (m_k, \Theta(m_k,\mathbf{m}))  \}
\]
as the macro-control-decision made by IS-based supervisor  at  macro-state $\textbf{m}$.

Initially, the state-estimate of the supervisor is $m_0=\{x_0\}$ and the intruder believes that this is the unique estimate of the system  with estimate  $\mathbf{m}_0=\{  m_0  \}$, which forms the initial information-state $\imath_0=(m_0,\mathbf{m}_0)$.

Then the supervisor issues a non-deterministic  decision set  $\Gamma_0=S_N(\epsilon)=\Theta( m_0,\textbf{m}_0)$.
Note that, we have pre-specified that the supervisor is IS-based.
Therefore, we denote the control decision information at this instant by a macro-control-decision
$d_{\Theta}(\textbf{m}_0)=\{(m_0, \Theta( m_0,\textbf{m}_0) )  \}$, which means that ``\emph{the supervisor will make control decision if its state-estimate is $m_0$}".
Note that, at this instant,   $d_{\Theta}(\textbf{m}_0)$ is a singleton as the intruder does not yet have ambiguity about   the supervisor, i.e., $\textbf{m}_0=\{ m_0\}$.

Once the allowed decision set $\Gamma_0$ is specified, the supervisor will pick a concrete control decision in it.
The intruder does not know which decision is chosen while the supervisor knows.
Suppose that $\Gamma_0=\{\gamma_0^1,\dots,\gamma_0^k\}$ contains $k$ control decisions.
Then the intruder's knowledge about the supervisor becomes
\begin{align}\label{update1}
\mathbf{m}_0^+
=& \odot( d_{\Theta}(\textbf{m}_0)  ) \\
=&\{ (U\!R_{\gamma_0^1}( m_0), \gamma_0^1), \dots, (U\!R_{\gamma_0^k}( m_0), \gamma_0^k)  \}  \nonumber  \\
=&\{(m_0^1,\gamma_0^1), \dots, (m_0^k,\gamma_0^k)\}, \nonumber
\end{align}
which means that the supervisor's estimate (with the unobservable tail) is possibly $U\!R_{\gamma_0^i}( m_0)$ and the control decision applied is $\gamma_0^i$.
Note that, the supervisor knows precisely which augmented micro-state $(m_0^i,\gamma_0^i)$ it is at.

Then when a new observable event $\sigma\in \Sigma_o$ occurs, and  the intruder updates its knowledge   to
\begin{align}\label{update2}
\mathbf{m}_1
=  \widehat{N\!X}_{\sigma}( \mathbf{m}_0^+   )
= \{m_1^1, \dots, m_1^{p}\}.
\end{align}
which is a macro-state containing at most $k$ micro-states, i.e. $p\leqslant k$.

Now, let us assume that, after some steps, the intruder's knowledge about the supervisor (immediately after the occurrence of an observable event) is
\[
\mathbf{m}_n = \{m_n^1, \dots, m_n^{k}\},
\]
Note that the supervisor knows the exact state estimate, i.e., $m_n^i\in \mathbf{m}_n$,  and for each $m_n^i$,
it allows  non-deterministic decision set $\Gamma_i=\Theta( m_n^i,\mathbf{m}_n )$ as we assume the supervisor is IS-based and is encoded by $\Theta$.
Therefore, the corresponding macro-control-decision is
\begin{align}\label{d}
d_{\Theta}(\textbf{m}_n)\!\!=\!\!\{\!(m_n^1, \Theta( m_n^1,\mathbf{m}_n )  ),\dots,  ( m_n^k, \Theta( m_n^k,\mathbf{m}_n )  ) \!\}.
\end{align}
Then the intruder's knowledge about the supervisor is updated by adding this control information
\[
\mathbf{m}_n^+ = \odot( d_{\Theta}(\textbf{m}_n)  ),
\]
which is an augmented marco-state containing at most $\sum_{i=1}^k |\Theta( m_n^i,\mathbf{m}_n )|$ augmented micro-states.

Based on the above discussion,
suppose that the intruder observes $\sigma_1\cdots\sigma_n\in P(\mathcal{L}(S_N/G))$ and by assuming the fact that $S_N$ is an IS-based supervisor encoded by $\Theta$,
it induces the following sequence
\begin{align}\label{eq:induce}
\textbf{m}_0
\xrightarrow{d_0} \textbf{m}_0^+
\xrightarrow{\sigma_1} \textbf{m}_1
\xrightarrow{d_1}
\dots
\xrightarrow{\sigma_n}
\textbf{m}_n
\xrightarrow{d_n} \textbf{m}_n^+,
\end{align}
where
$\textbf{m}_0=\{ \{x_0\} \}$, $d_i= d_{\Theta}(\textbf{m}_i)$, $\mathbf{m}_i^+ = \odot(d_i)$ and 	
$\mathbf{m}_{i+1} =  \widehat{N\!X}_{\sigma_{i+1}}( \mathbf{m}_i^+   )$.
We note that $\sigma_{i+1}$ is defined at  $\textbf{m}_i^+$ iff
there exist $(m,\gamma)\!\in\! \textbf{m}_i^+$ and $x\in m$ such that $\delta(x,\sigma_{i+1})! $ and  $\sigma_{i+1}\!\in\! \gamma  $.
Therefore, the sequence in Equation~\eqref{eq:induce} is uniquely defined when $\sigma_1\cdots\sigma_n$ and $\Theta$ are fixed;
it is independent from the actual online choice of the supervisor at each instant.

Now we are ready to specify the reachability closure of an IS-mapping.
Formally let $\Theta: I\to 2^\Gamma$ be a partial IS-mapping and $\imath=(m,\mathbf{m})\in I$ be an information-state.
Then the \emph{reachability closure} of $\imath$ under $\Theta$,
denoted by $\textsc{Reach}_{\Theta}( \imath )\subseteq I$, is defined recursively as follows:
\begin{itemize}
	\item
	$\imath\in  \textsc{Reach}_{\Theta}( \imath )$;
	\item
	$\imath'=(m',\mathbf{m}') \in  \textsc{Reach}_{\Theta}( \imath )$  if
	\begin{itemize}
		\item
		$m'\in \mathbf{m}'$; and
		\item
		there exists $\imath''=(m'',\mathbf{m}'')\in   \textsc{Reach}_{\Theta}( \imath )$
		such that $\textbf{m}'' \xrightarrow{  d_{\Theta}( \textbf{m}''  )    } \textbf{m}''^+
		\xrightarrow{\sigma} \textbf{m}'$
		for some $\sigma\in \Sigma_o$.
	\end{itemize}
\end{itemize}

\subsection{Property of the Information-State}
The above analysis of information-flow is heuristic.
In this subsection, we formally show that the proposed information updating rule indeed yields the state estimate of the intruder in the controlled system.
\begin{mythm}\label{thm1}
	Let $\Theta$ be an IS-mapping that encodes an IS-based  supervisor $S_N$
	and $s=\sigma_1\dots\sigma_k \in P(\mathcal{L}(S_N/G))$ be an observable string available to the intruder.
	Let $\textbf{m}_k$ and $\textbf{m}_k^+$ be states induced by $s$ and $\Theta$ according to Equation~\eqref{eq:induce}.
	Then we have
	\begin{itemize}
		\item[(i)]
		$\textbf{m}_k=\hat{\mathcal{E}}_I(s)$;
		and
		\item[(ii)]
		\begin{equation}
		\textbf{m}_k^+=
		\left\{
		(\mathcal{E}_S(\rho\gamma),\gamma):
		\begin{gathered}
		\rho\in\mathcal{O}(\mathcal{L}_e^o(S_N/G))\text{ s.t. }\\
		\rho|_\Sigma=s \text{ and } \gamma\in S_N( \rho )
		\end{gathered}
		\right\}.\nonumber
		\end{equation}
	\end{itemize}
\end{mythm}

\begin{proof}
	We prove   by induction on the length of $s$.
	
	\emph{Induction Basis: }
	For $|s|=0$, i.e. $s=\epsilon$, from the definition of $\hat{\mathcal{E}}_I(s)$,  we know that
	\[
	\begin{aligned}
	\hat{\mathcal{E}}_I(\epsilon)=&\{ \hat{\mathcal{E}}_S ( \rho )\in 2^X :   \rho\in \mathcal{O}(\mathcal{L}_e^o(S_N/G)) \text{ s.t. }   \rho|_{\Sigma}=\epsilon   \}\\
	=&\{\hat{\mathcal{E}}_S(\epsilon)\}\\
	=&\{\{\delta(\epsilon)\}\}\\
	=&\{\{x_0\}\}\\
	=&\mathbf{m}_0
	\end{aligned}
	\]
	Since
	$\mathbf{m}_0^+=\odot( d_{\Theta}(\textbf{m}_0))$ and   $d_{\Theta}(\mathbf{m}_0)= \{(m_0,S_N(\epsilon))\}  $,
	we have
	\begin{equation}
	\begin{aligned}
	\textbf{m}_0^+=&\odot(d_{\Theta}(\mathbf{m}_0))\\
	=&\{(U\!R_\gamma(m_0),\gamma):\gamma\in  S_N(\epsilon) \}\\
	=&\left\{(U\!R_\gamma(\hat{\mathcal{E}}_S(\epsilon)),\gamma):
	\gamma\in  S_N(\epsilon)
	\right\}\\
	=&\left\{
	(\mathcal{E}_S(\gamma),\gamma):
	\gamma\in  S_N(\epsilon)
	\right\}
	\nonumber
	\end{aligned}
	\end{equation}
	Note that $\rho=\epsilon$ is the only extended string in 	$\mathcal{O}(\mathcal{L}_e^o(S_N/G))$ such that  $\rho|_\Sigma=s$.
	This completes the induction basis.
	
	\emph{Induction Step: }
	Let us assume that Theorem~\ref{thm1} holds for $|s|=k$.
	Then we want to prove the case of $s\sigma_{k+1}\in P(\mathcal{L}(S_N/G))$.
	By the induction hypothesis, we know that
	\begin{equation}
	\textbf{m}_k^+=
	\left\{
	(\mathcal{E}_S(\rho\gamma),\gamma):
	\begin{gathered}
	\rho\in\mathcal{O}(\mathcal{L}_e^o(S_N/G))\text{ s.t. }\\
	\rho|_\Sigma=s \text{ and } \gamma\in S_N( \rho )
	\end{gathered}
	\right\}.\nonumber
	\end{equation}
	Then, we have
	\[
	\begin{aligned}
	&\mathbf{m}_{k+1}\\
	=&\widehat{N\!X}_{\sigma_{k+1}}(\mathbf{m}_k^+)\\
	=&\{ N\!X_{\sigma_{k+1}}(m): (m,\gamma)\in \mathbf{m}_k^+, \sigma_{k+1}\in \gamma \}\\
	=&\left\{
	N\!X_{\sigma_{k+1}}(\mathcal{E}_S(\rho\gamma)):
	\begin{gathered}
	\rho\in\mathcal{O}(\mathcal{L}_e^o(S_N/G))\text{ s.t. }\\
	\rho|_\Sigma=s, \gamma\in S_N( \rho )\text{ and }\sigma_{k+1}\in \gamma
	\end{gathered}
	\right\}\\	
	=&\{\hat{\mathcal{E}}_S(\rho\gamma\sigma_{k+1}):\rho\in\mathcal{O}(\mathcal{L}_e^o(S_N/G))\text{ s.t. }
	\rho|_\Sigma=s\}\\
	=&\{\hat{\mathcal{E}}_S(\rho'):\rho'\in\mathcal{O}(\mathcal{L}_e^o(S_N/G))
	\text{ s.t. }\rho'|_\Sigma=s\sigma_{k+1}\}\\
	=&\hat{\mathcal{E}}_I(s\sigma_{k+1}).
	\end{aligned}
	\]
	For $\mathbf{m}_{k+1}^+=\odot( d_{\Theta}(\textbf{m}_{k+1}))$.
	Suppose that  $d_{\Theta}(\textbf{m}_{k+1})=\{(m_{k+1}^1,\Theta(m_{k+1}^1,\textbf{m}_{k+1})),\ldots,(m_{k+1}^n,\Theta(m_{k+1}^n,\textbf{m}_{k+1}))\}$.
	Note that we have
	\begin{align}
	\textbf{m}_{k+1}
	=&\{\hat{\mathcal{E}}_S(\rho):\rho\in\mathcal{O}(\mathcal{L}_e^o(S_N/G))
	\text{ s.t. }\rho|_\Sigma=s\sigma_{k+1}\}\nonumber\\
	=&\{m_{k+1}^1,\dots,m_{k+1}^n\}\nonumber
	\end{align}
	For each $\rho\in\mathcal{O}(\mathcal{L}_e^o(S_N/G))$  such that
	$\rho|_\Sigma=s\sigma_{k+1}$, since $S_N$ is IS-based, we have
	$S_N(\rho)=\Theta( \hat{\mathcal{E}}_S(\rho),\textbf{m}_{k+1} )$.
	Then we have the followings
	\begin{equation}
	\begin{aligned}
	&\mathbf{m}_{k+1}^+\\
	=&\odot( d_{\Theta}(\textbf{m}_{k+1}))\\
	=&\{(U\!R_\gamma(m),\gamma):\exists (m,\Gamma)\in d_{\Theta}(\textbf{m}_{k+1})\text{ s.t. } \gamma\in\Gamma\}\\
	=&\left\{(U\!R_\gamma(\hat{\mathcal{E}}_S(\rho)),\gamma):
	\begin{gathered}
	\rho\in\mathcal{O}(\mathcal{L}_e^o(S_N/G))\text{ s.t. }\\
	\gamma\in S_N(\rho) \text{ and }
	\rho|_\Sigma=s\sigma_{k+1}
	\end{gathered}
	\right\}\\
	=&\left\{
	(\mathcal{E}_S(\rho\gamma),\gamma):
	\begin{gathered}
	\rho\in\mathcal{O}(\mathcal{L}_e^o(S_N/G)),\text{ s.t. }\\
	\gamma\in S_N(\rho)\text{ and }
	\rho|_\Sigma=s\sigma_{k+1}
	\end{gathered}
	\right\}
	\end{aligned}\nonumber
	\end{equation}
	This completes the induction step, i.e. (ii) holds.
\end{proof}

For any augmented macro-state $\textbf{m}^+$,
we define
\[
M(\textbf{m}^+)=\{ m\in M: (m,\gamma)\in \textbf{m}^+  \}
\]
as the macro-state obtained by removing the control decision components from $\textbf{m}^+$.
Then the following result reveals that the above defined states set
$M(\textbf{m}_k^+)$ are indeed the state estimate of the intruder ${\mathcal{E}}_I(s)$.
\begin{mycol}\label{corollary1}
	Let $\Theta$ be an IS-mapping that encodes an IS-based  supervisor $S_N$ and $s=\sigma_1\ldots\sigma_k\in P(\mathcal{L}(S_N/G))$ be an observable string available to the intruder.
	Let $\textbf{m}_k^+$ be the state reached according to the information-flow.
	Then we have
	\[
	M(\textbf{m}_k^+)={\mathcal{E}}_I(s)
	\]
\end{mycol}
\begin{proof}
	By Theorem \ref{thm1}, we have
	\begin{equation}
	\textbf{m}_k^+=
	\left\{
	(\mathcal{E}_S(\rho\gamma),\gamma):
	\begin{gathered}
	\rho\in\mathcal{O}(\mathcal{L}_e^o(S_N/G))\text{ s.t. }\\
	\rho|_\Sigma=s \text{ and } \gamma\in S_N( \rho )
	\end{gathered}
	\right\}.\nonumber
	\end{equation}
	Therefore,
	\begin{equation}
	\begin{split}
	M(\textbf{m}_k^+)=	
	&\left\{
	\mathcal{E}_S(\rho\gamma):
	\begin{gathered}
	\rho\in\mathcal{O}(\mathcal{L}_e^o(S_N/G))\text{ s.t. }\\
	\rho|_\Sigma=s \text{ and } \gamma\in S_N( \rho )
	\end{gathered}
	\right\}\\
	=&\{\mathcal{E}_S(\rho'):\rho'\in\mathcal{O}(\mathcal{L}_e^d(S_N/G))\text{ s.t. }\rho'|_\Sigma=s\}\\
	=&\mathcal{E}_I(s)
	\nonumber
	\end{split}
	\end{equation}
	This completes the proof.
\end{proof}

We explain the above concepts by the following example.

\begin{myexm}\upshape
Let us still consider system $G$ in Fig.\ref{system}.
We consider a non-deterministic supervisor $S_N$ defined by
\[
\forall \rho\in (\Gamma\Sigma_o)^*: S_N(\rho)=\{\{c_1\},\{c_2\}\}.
\]
Clearly, this supervisor is IS-based and it can be encoded by IS-mapping
$\Theta: I\to 2^\Gamma$ such that $\forall \imath\!\in\! I: \Theta(\imath)\!=\!\{\{c_1\},\{c_2\}\}$.

Initially, the  supervisor's estimate is $m_0=\{0\}$ and
the intruder's estimate of supervisor's estimation is $\mathbf{m}_0=\{\{0\}\}$, where the macro-control-decision induced by $\Theta$ is
\[
d_{\Theta}(\mathbf{m}_0)=\{  (\{0\},\Theta( \{0\},\{\{0\}\}  ))        \}=\{(\{0\},\{\{c_1\},\{c_2\}\})\}.
\]
Then  the intruder's knowledge   is updated to
\begin{align}
   \mathbf{m}_0^+
&= \odot( d_{\Theta}(\mathbf{m}_0)   ) \nonumber\\
&= \{ ( U\!R_{   \{c_1\}    }  (  \{0\}   ) ,\{c_1\} )  , ( U\!R_{   \{c_2\}    }  (  \{0\}   )     ,\{c_2\}  )   \}\nonumber\\
&= \{(\{0,1\},\{c_1\}),(\{0,3\},\{c_2\})\}.\nonumber
\end{align}
When  event $o_1$ is observed,
the intruder updates its knowledge  to
\begin{align}
   \mathbf{m}_1
&=\widehat{N\!X}_{o_1}(\mathbf{m}_0^+)  = \{ N\!X_{o_1}( \{0,1\}  )  ,  N\!X_{o_1}( \{0,3\}  )  \}\nonumber\\
&= \{  \{4\} ,  \{5\}  \},\nonumber
\end{align}
which means that the intruder guesses that the supervisor's state-estimate is either $\{4\}$ or $\{5\}$ based on the information available.
Again, the macro-control-decision at $\mathbf{m}_1$ is
\begin{align}
d_{\Theta}(\mathbf{m}_1)
=&\{  (\{4\},\Theta( \{4\}, \mathbf{m}_1  ))  , (\{5\},\Theta( \{4\}, \mathbf{m}_1  ))\} \nonumber\\
=&\{(\{4\},\{\{c_1\},\{c_2\}\}),(\{5\},\{\{c_1\},\{c_2\}\})\}, \nonumber
\end{align}
which leads to
\begin{align}
\mathbf{m}_1^+
&= \odot( d_{\Theta}(\mathbf{m}_1)   ) \nonumber\\
&=
\left\{
\begin{array}{l l}
 ( U\!R_{   \{c_1\}    }  (  \{4\}   ) ,\{c_1\} )  , ( U\!R_{   \{c_2\}    }  (  \{4\}   )     , \{c_2\}  ),\\
( U\!R_{   \{c_1\}    }  (  \{5\}   ) ,\{c_1\} )  , ( U\!R_{   \{c_2\}    }  (  \{5\}   )     ,\{c_2\}  )
\end{array}
\right\}\nonumber\\
&= \{\!(\{4\},\!\{c_1\}),(\{4\},\!\{c_2\}),(\{5,6\},\!\{c_1\}),(\{5,7\},\!\{c_2\})\!\}.\nonumber
\end{align}
Similarly, from $\mathbf{m}_1^+$, observations can be observed and so forth.
\end{myexm}

\subsection{Decode Supervisor from IS-Mapping}

Finally, we are ready to discuss how to \emph{decode} an IS-based non-deterministic supervisor from an IS-mapping $\Theta: I\to 2^\Gamma$.
The decoded non-deterministic supervisor   is denoted by $S_\Theta$.
Let $\textsc{Dom}(\Theta)=\{  \imath\in I: \Theta( \imath )!\}$ be the domain of $\Theta$.
We say that IS-mapping $\Theta$ is \emph{reachability-closed} if
\[
\textsc{Reach}_{\Theta}(\imath_0) \subseteq \textsc{Dom}(\Theta),
\]
where  $\imath_0=( \{x_0\},\{ \{x_0\}  \} )$ is the initial information-state.
Clearly, $\Theta$ is necessarily to be reachability-closed;
otherwise, the supervisor cannot make decision after some executions.
Without loss of generality, we can further assume that $\textsc{Reach}_{\Theta}(\imath_0) = \textsc{Dom}(\Theta)$ as the mapping information of those unreachable states are not used.

When  IS-mapping $\Theta$ is reachability-closed,  we can decode a supervisor $S_\Theta$ as follows.
For any decision history $\rho=\gamma_0\sigma_1\gamma_1\dots  \gamma_{n-1}\sigma_n $, we have
\begin{equation}
S_\Theta(\rho)= \Theta(   \hat{\mathcal{E}}_{S}(\rho) , \hat{\mathcal{E}}_I(\rho|_{\Sigma})         ).
\end{equation}
Note that, based on the previous discussion,  both  $\hat{\mathcal{E}}_{S}(\rho)$  and $\hat{\mathcal{E}}_I(\rho|_{\Sigma}) $ can be computed recursively based on $\Theta$.
Therefore, in practice,   $S_\Theta(\rho)$ can be executed online according to Algorithm~1.
Specifically, we use parameters
$m,m^+,\mathbf{m}$ and $ \mathbf{m}^+ $ to represent
$\hat{\mathcal{E}}_{S}(\rho),  {\mathcal{E}}_{S}(\rho), \hat{\mathcal{E}}_{I}(\rho|_{\Sigma}) $ and ${\mathcal{E}}_{I}(\rho|_{\Sigma})$, respectively.
Note that the updates of $\mathbf{m}$ and $\mathbf{m}^+$ use  the online observation $\sigma$ and the IS-mapping $\Theta$ to generate a non-deterministic control decision set $\Gamma$, in which an actual control decision applied $\gamma\in \Gamma$ is chosen randomly.
However,  the updates of $m$ and $m^+$ only use  the online observation $\sigma$ and the  actual decision $\gamma$ applied.

\begin{algorithm}
	\SetKwData{Left}{left}
	\SetKwData{Up}{up}
	\SetKwFunction{FindCompress}{FindCompress}
	\SetKwInOut{Input}{input}
	\SetKwInOut{Output}{output}
	
	\Indp
	\BlankLine
	\nl   $m \gets\{x_0\}$ and $\mathbf{m}\gets \{\{x_0\}\}$ and $\rho\gets \epsilon$\;
	\nl  \While{$\rho=\epsilon$ or a new event $\sigma\in \Sigma_o$ is observed  }
	{
	\If{a new event $\sigma\in \Sigma_o\cap \gamma$ is observed}
	{
	\nl   $m\gets N\!X_{\sigma}( m^+  )$ and $\mathbf{m}\gets  \widehat{N\!X}_{\sigma}(\mathbf{m}^+)$\;
	\nl   $\rho\gets \rho\sigma$\;
    }	
	\nl
	Define $S_\Theta(\rho)\gets \Theta( m ,\mathbf{m})$ as the current non-deterministic control decision\;
	\nl
	Randomly pick $\gamma\in S_\Theta(\rho)$ and apply this control decision online\;
	\nl
	$ m^+ \gets U\!R_{\gamma}( m  )$ and
	$\mathbf{m}^+\gets \odot(d_{\Theta}(\mathbf{m}))$\;	
	\nl   $\rho\gets \rho\gamma$\;
	}
	\caption{Online Decoding of IS-Mapping $\Theta$\label{constructS}}
\end{algorithm}

By understanding how an IS-mapping $\Theta$ can be decoded as an IS-based supervisor,
hereafter, we will also refer to a reachability-closed  IS-mapping $\Theta$ as an IS-based supervisor directly.
In order to solve the general opacity enforcement problem as formulated in Problem~1,  our approach is to  restrict  our solution space to IS-based supervisors and solve the following IS-mapping synthesis problem.

\begin{myprob}\label{prob2}
(Information-State-Based Opacity Enforcement Problem)
Given system $G$ and secret states $X_S\subseteq X$, synthesize an IS-based supervisor $S_\Theta: (\Gamma\Sigma_o)^* \rightarrow 2^\Gamma$ decoded from IS-mapping $\Theta: I\rightarrow 2^\Gamma$, such that $S_\Theta /G$ is opaque w.r.t. $X_S$ and $\Sigma_o$.		
\end{myprob}

\begin{remark}\upshape
Problem 2 essentially restricts the solution space of Problem~1 to a finite domain.
Clearly, if there exists an IS-based supervisor that enforces opacity, then there exists a non-deterministic opacity-enforcing supervisor. However, the following question arise immediately:
\emph{whether or not the non-existence of an IS-based supervisor also implies the non-existence of a general supervisor?}
We will show later in Section~\ref{section6} that   there exists a  non-deterministic opacity enforcing supervisor
	\emph{if and only if} there exists an IS-based one.	In other words, restricting our attention to Problem~2 is without loss of generality for the solvability of Problem~1.
\end{remark}

\section{Synthesis of IS-Based Supervisors}\label{section5}
In this section, we discuss how to synthesize an IS-based supervisor that enforces opacity.
We first introduce the structure of the generalized bipartite transition system.
Then we present a synthesis algorithm that returns a solution to Problem \ref{prob2}.

\subsection{Bipartite Transition System}
By the  analysis in the previous section, we see that the update of the intruder's knowledge consists of two steps:
one is when the supervisor picks a macro-control-decision and the other is when a new observable event occurs.
To separate these two steps, we  adopt the idea of the bipartite transition systems (BTS) proposed in \cite{yin2016uniform}.
Here, we call the proposed structure \emph{generalized} BTS (G-BTS) as it captures, in a more general manner,  both the supervisor's estimate and the intruder's knowledge about the supervisor, while the original BTS in \cite{yin2016uniform} only captures the supervisor's estimate.

\begin{mydef}
	A generalized bipartite transition system (G-BTS) $T$ w.r.t.\ $G$ is a $7$-tuple
	\[
	T=(Q_Y,Q_Z,h_{Y\!Z},h_{ZY},\Sigma_o,D,y_0).
	\]
	where
	\begin{itemize}
		\item
		$Q_Y\subseteq \mathbb{M}$ is a set of macro-states;
		\item
		$Q_Z\subseteq \mathbb{M}^+$ is the set of augmented macro-states;
		\item
		$h_{Y\!Z}:Q_Y\times D\rightarrow Q_Z$ is the transition function from $Y$-states to $Z$-states satisfying:
		for any $h_{Y\!Z}(\mathbf{m},d)=\mathbf{m}^+$, we have
		\begin{itemize}
			\item
			$d$ is \emph{compatible} with $\mathbf{m}$; and
			\item
			$\mathbf{m}^+=\odot(d)$.
		\end{itemize}
		\item
		$h_{ZY}:Q_Z\times\Sigma_o \rightarrow Q_Y$ is the transition function from $Z$-states to $Y$-states satisfying:
		for any $h_{ZY}(\mathbf{m}^+,\sigma)=\mathbf{m}$, $\sigma\in\Sigma_o$, we have
		\begin{itemize}
			\item
			$\mathbf{m} = \widehat{N\!X}_{\sigma}( \mathbf{m}^+   )$.
		\end{itemize}
		\item
		$D$ is the set of macro-control-decisions;
		\item
		$\Sigma_o$ is the set of observable events of system $G$;
		\item
		$y_0=\{\{x_0\}\} \in Q_Y$ is the initial $Y$-state.
	\end{itemize}
\end{mydef}

The G-BTS essentially captures the information-flow analyzed in Section~\ref{section4}.
Specifically, at each $Y$-state, the IS-based supervisor makes  a macro-control-decision $d$ and then moves to a $Z$-state by updating the intruder's knowledge via unobservable reaches under the issued macro-control-decision $d$.
When a new observable event $\sigma\in \Sigma_o$ occurs at a $Z$-state, we move to a $Y$-state by computing the observable reach, and so forth.

\begin{figure*}[htbp]
	\centering
	\includegraphics[width=1\linewidth]{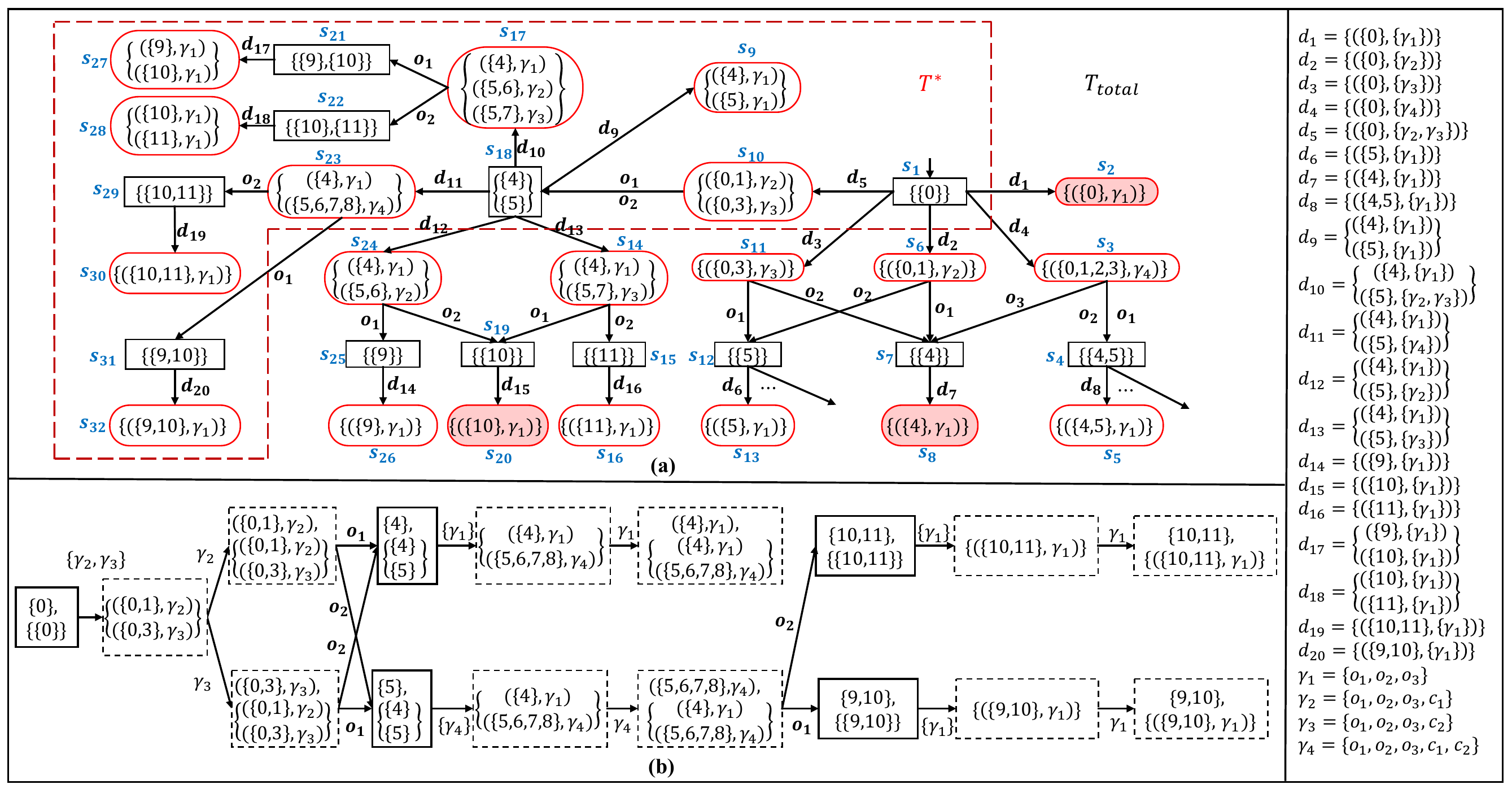}
	\caption{(a) An example of the  G-BTS, where		
		rectangular  states represent   $Y$-states and oval states represent $Z$-states.			
		(b) Decision diagram of the synthesized non-deterministic supervisor.
	}
	\label{G-BTS}
\end{figure*}

\begin{myexm}\label{exm3}\upshape
	Again, we consider system $G$ in Fig.~\ref{system}.
	An example of the G-BTS is shown in Fig.~\ref{G-BTS}(a),
	in which rectangular states represent $Y$-states and oval states represent $Z$-states.
    Some states are omitted in Fig.~\ref{G-BTS}(a) for simplicity.
	States are named by $s_1,\dots,s_{32}$.
	The initial $Y$-state is $s_1=\{\{0\}\}$, from which macro-control-decisions $d_1,\dots,d_5$ that are compatible with $s_1$ can be made.
	For example, if the macro-control-decision made is $d_5=\{  (\{0\} ,\{   \{c_1 \},\{c_2 \} \} )    \}$,
	then we move to $Z$-state $s_{10}=\odot(d_5)$.
	From this state,  observable events $o_1$ and $o_2$ can occur, and both lead to the same $Y$-state $s_{18}$.
    From $Y$-state $s_{18}$, macro-control-decisions $d_9,\ldots, d_{13}$ that are compatible with $s_{18}$ can be made.
    If the macro-control-decision made is $d_{11}$, then we move to $s_{23}=\odot(d_{11})$ and so forth.
\end{myexm}

\subsection{Synthesis of IS-Based Supervisors}

Now, we present how to synthesize IS-based non-deterministic opacity-enforcing supervisors represented by an IS-mapping.
Given a G-BTS $T$, for any $Y$-state $y\in Q_Y$,
we define
\[
C_T(y) := \{d\in D :h_{Y\!Z}(y,d)!\}
\]
as the  set of macro-control-decisions defined at $y$ in $T$.
Also, we say that
a $Y$-state $y$ is \emph{consistent} if $C_T(y)\neq \emptyset$;  and
a $Z$-state $z$ is \emph{consistent} if, for any $\sigma\in \Sigma_o$, we have
\[
h_{ZY}(z,\sigma)!\Leftrightarrow   (\exists (m,\gamma)\!\in\! z )[  {N\!X}_\sigma(m)\neq \emptyset \wedge \sigma\!\in\! \gamma  \}].
\]
Intuitively, a $Y$-state is consistent if at least one macro-control-decision is defined and a $Z$-state is consistent if all feasible events are defined.
Consistency is required for the purpose of control as the supervisor should be able to provide a control decision for any observation.
We denote by $Q_{const}^T$ the set of consistent states    in $T$
and we say that $T$  is consistent if all reachable states are consistent.

As  discussed earlier, we restrict our attention to IS-based supervisors.
Our approach for synthesizing non-deterministic opacity-enforcing supervisors consists of the following two steps:
\begin{enumerate}[(i)]
	\item
	construct the largest  consistent G-BTS  in which all states are not secret-revealing;
	\item
	extract one IS-based supervisor in the form of an IS-mapping from this largest G-BTS.
\end{enumerate}
Since such an IS-based supervisor is extracted from $T$,  by Theorem~\ref{thm1} and Corollary \ref{corollary1},
we know that, upon the occurrence of any decision history, the $Z$-state $z\in \mathbb{M}^+$ reached is essentially the set of all possible state-estimates of the supervisor.
Moreover, by  Proposition \ref{pro1}, we know that $\bigcup M(z) = X_I(s)$, where $s$ is the observation leading to the $Z$-state.
Therefore,  to make sure that the  closed-loop system $S_\Theta/G$ is opaque,
it suffices to guarantee that, for any $Z$-state  $z\in \mathbb{M}^+$ reached, we have
\[
\bigcup M(z)\nsubseteq X_S.
\]
To this end, we define
\[
Q_{reveal}=\{z\in \mathbb{M}^+:  \bigcup M(z)\subseteq X_S\}
\]
as the set of \emph{secret-revealing} $Z$-states.

In order to synthesize  an IS-based supervisor,
first, we  construct the largest G-BTS w.r.t.\ $G$ that enumerates all the feasible transitions satisfying the constraints of $h_{Y\!Z}$ and $h_{ZY}$.
We  denote  such an all-feasible G-BTS by $T_{total}$.
Then, we need to delete all   secret-revealing $Z$-states in $T_{total}$ and obtain a new G-BTS
\[
T_0=T_{total} \!\!\restriction\!_{ (Q_Y\cup Q_Z)\setminus Q_{reveal}},
\]
where $T\!\!\restriction\!\!_Q$ denotes the G-BTS obtained by restricting the state-space of $T$ to  $Q\subseteq Q_Y\cup Q_Z$.

However, by deleting secret-revealing states, the resulting G-BTS may become inconsistent.
Hence, we also need to delete   inconsistent states recursively.
Specifically, we define an operator $F$ that maps a G-BTS to a new G-BTS by:
\[
F: T \mapsto T \!\!\restriction\!\!_{Q_{const}^T},
\]
and we define
\[
T^*=\lim_{k\to\infty}F^k(T_0)
\]
as the largest consistent G-BTS in which there is no secret-revealing state.
The existence of the  supremal fixed-point as well as the finite-convergence of iteration follow directly from the computation of winning region in two-player games \cite{gradel2002automata} or the well-known supremal controllable sub-language \cite{Lbook}.

The construction of $T^*$ follows directly from its definition and one can proceed in two steps.
First, we construct $T_0$ by a depth-first search or a breadth-first search  from  the initial $Y$-state $y_0$.
Specifically, at each state encountered, one needs to consider all possible successor states, until reaching a secret-revealing state  or a state that has been visited.  Second,  we  prune inconsistent states from $T_0$  by iterations.
Specifically, we need to remove $Y$-states having no successor and $Z$-states at which some feasible transitions are undefined, until the structure converges.
Similar searching and pruning procedure can  be found in the literature; see, e.g., Algorithm~1 in  \cite{yin2016uniform}.
We illustrate this procedure by the following example.

\begin{myexm}\upshape
Consider again system $G$ in Fig.~\ref{system}.
First, we construct the largest G-BTS $T_{total}$ by enumerating all possible transitions, which is in fact the structure shown in Fig.~\ref{G-BTS}(a).
For sake of simplicity, as we discussed earlier, redundant  macro-control-decisions are omitted in $T_{total}$.
For example,  $d=\{(\{0\},\{\gamma_1,\gamma_2\})\}$ is not listed at state $s_1$, since $\gamma_1 \subset \gamma_2$ and macro-control-decision $d_2$ is sufficient enough to carry this information.

Note that $Z$-states $s_2$, $s_8$ and $s_{20}$ are secret-revealing states since $\bigcup M(s_2)=\{0\}\subseteq X_S$, $\bigcup M(s_{8})=\{4\}\subseteq X_S$ and $\bigcup M(s_{20})=\{10\}\subseteq X_S$.
Therefore, we need to delete states $s_2$, $s_8$ and $s_{20}$ to obtain $T_0$.
However, this creates inconsistent states $s_7$ and $s_{19}$ as no macro-control-decision is defined.
Therefore, these two states are removed when applying operator $F$ for the first time.
Again, this further creates inconsistent states $s_3$, $s_6$, $s_{11}$, $s_{14}$ and $s_{24}$ since some feasible observations are not defined.
Therefore, these states and the associated transitions are again deleted when applying operator $F$ for the second time.
This yields the final structure $T^*$  including states  $s_1$,$s_9$,$s_{10}$, $s_{17}$, $s_{18}$, $s_{21}$, $s_{22}$, $s_{23}$, $s_{27},\ldots,s_{32}$,
which is the largest consistent G-BTS having no secret-revealing state.
\end{myexm}

\begin{algorithm}[htbp]
	\SetKwData{Left}{left}
	\SetKwData{Up}{up}
	\SetKwFunction{FindCompress}{FindCompress}
	\SetKwInOut{Input}{input}
	\SetKwInOut{Output}{output}
	
	\Input{$T^*$}
	\Output{$\Theta^*$}
	\Indp
	\BlankLine
	\nl
	$y_0\gets \{\{x_0\}\}$, \textsc{Visited}$\gets \{y_0\}$, $\textsc{Dom}(\Theta^*)\gets\emptyset$  \;
	\nl\If{$y_0$ is not in $T^*$}
	{
	\nl \textbf{return} ``no solution" \;
    }
	\nl
	{DoDFS}($\Theta^*, y_0   $)\;
	\nl
	\textbf{return} {$\Theta^*$}\;
	\BlankLine
	\Indm
	{\bf procedure} DoDFS($\Theta^*, y $)\;
	\Indp
	\nl
	 choose a macro-control-decision $d \in C_{T^*}(y)$ \;
	\nl
	\For{$m\in y$}
	{
	\nl $\textsc{Dom}(\Theta^*)\gets\textsc{Dom}(\Theta^*) \cup \{(m,y)\}  $\;
	\nl $\Theta^*(m,y)\gets \Gamma_{m,d}$,
	where $\Gamma_{m,d}$ is the unique decision such that $(m,\Gamma_{m,d}) \in d$\;
    }
	\nl
	\For{$\sigma\in \Sigma_o$ such that $h_{ZY}(h_{Y\!Z}(y,d) ,\sigma)=y'$}
	{
	\If{$y' \not\in \textsc{Visited}$}
	{
	\nl
	\textsc{Visited}$\gets$ \textsc{Visited}$\cup\{y'\}$  \;
	\nl
	DoDFS($\Theta^*, y'$)\;
	}
    }
	\BlankLine
	\Indm
	
	\caption{Synthesis of IS-Based Supervisor  $\Theta^*$\label{synthesis}}
\end{algorithm}

Based on $T^*$,    Algorithm~2 is provided to synthesize an IS-based non-deterministic supervisor in the form of an IS-mapping via a depth-first search.
Specifically, we start from the initial  $Y$-state and
pick a macro-control-decision $d$ from the set of all macro-control-decisions defined at $y$.
Then for each pair $(m,y)$, where $m\in y$, we use $d$ to define the mapping value for information-state $(m,y)$, which is the unique non-deterministic  decision set associated with $m$ in $d$.
Then we move to the unique $Z$-state reached under macro-control-decision $d$ and consider all successor $Y$-states by  considering all possible observable events.
If the new $Y$-state has not yet been visited, then we
repeat the selection procedure by making a recursive call of procedure DoDFS until all reachable information-states are traversed.
The computed IS-mapping $\Theta^*$ is reachability-closed by construction; hence can be used to decode a corresponding IS-based supervisor $S_{\Theta^*}$ according to Algorithm~1.

We show the computation procedure in Algorithm~2 by the following example.

\begin{myexm}\upshape
We still consider our running example with $T^*$ shown in Fig.~\ref{G-BTS} and we use Algorithm~2 to synthesize an IS-mapping from $T^*$. The algorithm starts from the initial $Y$-state $s_1=\{\{0\}\}$,
at which the macro-control-decision  in $T^*$ is unique.
Therefore, the  supervisor will pick $d_5$ which induces
partial mapping value $\Theta^*( \{0\}, \{\{0\}\} )=\{\gamma_2,\gamma_3\}$.
By choosing $d_5$, we move to $Z$-state $s_{10}$ and we need to consider all possible successor $Y$-states of $s_{10}$.
Here, both $o_1$ and $o_2$ from $s_{10}$ leads to $Y$-state $s_{18}=\{\{4\},\{5\}\}$, where three macro-control-decisions $d_{9}, d_{10}, d_{11}$ are defined.
Suppose that the supervisor chooses $d_{11}=\{ (\{4\},\{\gamma_1\}), (\{5\},\{\gamma_4\})    \}$.
This again induces
partial mapping values $\Theta^*( \{4\}, \{\{4\},\{5\}\} )=\{\gamma_1\}$
and $\Theta^*( \{5\}, \{\{4\},\{5\}\} )=\{\gamma_4\}$.
If observable event $o_2$ occurs, $Y$-state $s_{29}$ is reachable and the macro-control-decision defined is unique.
Therefore, by choosing $d_{19}$ at $s_{29}=\{\{10,11\}\}$,
partial mapping value
$\Theta^*( \{10,11\}, \{\{10,11\}\} )=\{\gamma_1\}$ is induced.
If observable event $o_1$ occurs, $Y$-state $s_{31}$ is reachable and the macro-control-decision defined is $d_{20}$.
Partial mapping value
$\Theta^*( \{9,10\}, \{\{9,10\}\} )=\{\gamma_1\}$ is induced.
This completes the construction of reachability-closed IS-mapping $\Theta^*$, which can also be represented as the decision diagram shown in Fig.~\ref{G-BTS}(b).
\end{myexm}

\begin{remark}\upshape
The main purpose of this paper is to synthesize a non-deterministic supervisor that guarantees opacity. Our focus is the solvability of this problem and whether or not the synthesized solution is maximally-permissive is out of the main scope of this work. Here, we provide a heuristic approach to improve the permissiveness of this solution. In line~6 of Algorithm~2, we do not put specific criterion for which macro-control-decision to choose from $C_{T^*}(y)$.
To  enhance the permissiveness of the supervisor, we can pick a \emph{locally maximal} macro-control-decision at each $Y$-state.
Formally, given two non-deterministic   decision sets $\Gamma_1$ and $\Gamma_2$, we denote
\begin{itemize}
	\item
	 by $\Gamma_1\leq \Gamma_2$ if
	$\forall \gamma\in   \Gamma_1 ,\exists\gamma'\in \Gamma_2: \gamma\subseteq \gamma'$;
	and
	\item
	by $\Gamma_1<\Gamma_2$ if
	$\Gamma_1\leq \Gamma_2 \text{ and }\exists \gamma\in   \Gamma_1,\exists\gamma'\in \Gamma_2: \gamma\subset  \gamma'$.
\end{itemize}
Then  for each $Y$-state $y=\{m_1,\ldots,m_k\}$ in $T^*$ and  two macro-control-decisions $d_1$ and $d_2$ defined at $y$, where $d_1=\{(m_1,\Gamma_1),\cdots,(m_k,\Gamma_k)\}$, $d_2=\{(m_1,\Gamma_1'),\cdots,(m_k,\Gamma_k')\}$,
we say $d_2$ is more permissive than $d_1$, denoted by
$d_1< d_2$ if
\begin{itemize}
	\item
	$\forall i\in \{1,\dots,k\},\Gamma_i\leq \Gamma_i'$; and
	\item
	$\exists  i\in \{1,\dots,k\}, \Gamma_i<\Gamma_i' $.
\end{itemize}
Therefore, in line~6 of Algorithm~2, one can choose  a locally maximal macro-control-decision $d\in C_{T^*}(y)$ in the sense of
\[
\forall  d'\in C_{T^*}(y): d\not< d'.
\]
For example,  for $T^*$  in Fig.~\ref{G-BTS}(a), there are three macro-control-decisions
$d_9=\{ (\{4\}  , \{\gamma_1\}   ) , (\{5\}  , \{\gamma_1\}  )   \},
d_{10}=\{ \!(\!\{4\}  , \!\{\gamma_1\}  \! ) , (\!\{5\}  ,\! \{\gamma_2,\gamma_3\}  \!)  \! \}$ and $d_{11}=\{\! (\!\{4\}  , \!\{\gamma_1\}  \! ) , (\!\{5\}  , \!\{\gamma_4\}  \!)   \!\}$ defined at $s_{18}$.
Then $d_{11}$ is a locally maximally macro-control-decision among these three.
For example, we have $d_{10}<d_{11}$ since
$\gamma_2\!\subset\! \gamma_4$ and $\gamma_3\!\subset\!\gamma_4$.
Therefore, for the sake of permissiveness, the IS-mapping synthesis procedure can pick $d_{11}$ instead of $d_{9}$ or $d_{10}$.
\end{remark}

We conclude this section by discussing the complexity of the proposed supervisor synthesis algorithm.
To construct the largest consistent G-BTS $T^*$, first, we need to build $T_{total}$, which contains at most $2^{2^{|X|}}$ $Y$-states and $2^{2^{|X|+|\Sigma_c|}}$ $Z$-states.
For each $Y$-state, there are at most $2^{|X|+|\Sigma_c|}$  transitions defined and
for each $Z$-state, there are at most $|\Sigma_o|$ transitions defined.
Overall, $T_{total}$ contains, in the worst-case, $2^{2^{|X|+|\Sigma_c|}}+ 2^{2^{|X|}}$ states and $2^{2^{|X|}}\cdot 2^{|X|+|\Sigma_c|}+ |\Sigma_o| \cdot 2^{2^{|X|+|\Sigma_c|}}$ transitions.
The complexity of removing all secret-revealing states to obtain G-BTS $T_0$ is linear in the size of $T_{total}$.
The complexity of removing all inconsistent states iteratively to obtain $T^*$ is quadratic in the size of $T_{total}$.
Once $T^*$ is constructed, we run Algorithm~2 to synthesize an IS-mapping $\Theta^*$, which is simply a depth-first search over the space of $T^*$ and the complexity is still linear in the size of $T^*$.
The resulting IS-mapping contains at most $2^{2^{|X|}}$ elements in its domain.
In order to execute the supervisor online, we use Algorithm~1 to decode IS-mapping $\Theta^*$.
To this end, the supervisor needs to store the IS-mapping   $\Theta^*$ computed offline, and during the online execution, record both the current state estimate $m$ and the current macro-state $\mathbf{m}$.
Note that $m$ contains at most $|X|$ states, while $\mathbf{m}$ contains at most $2^{|X|}$.
By making a new control decision upon the occurrence of a new observable event,  $m$ and $\mathbf{m}$ can be updated, respectively, in polynomial-time and exponential-time in the size of $G$.
Note that this online update transition can also be pre-computed offline and be stored as transition  rules together with the IS-mapping $\Theta^*$.
Overall, the entire complexity of the proposed synthesis approach is doubly-exponential in the size of the original plant, where the major complexity is spent for the offline computation.

\section{Properties of the Synthesis Procedure}\label{section6}
In this section, we formally prove the correctness of the synthesis procedure proposed in Section~\ref{section5}.
Note that in the formulation of  Problem~\ref{prob1}, supervisors make control decision based on the decision histories and can be non-IS-based in general.
However, our algorithm in Section~\ref{section5} solves a restricted version of Problem~1 by only considering IS-based supervisors as formulated in Problem~2.
Therefore, to show the correctness of the proposed synthesis procedure in the context of Problem~1, our arguments  consist of the following two steps:
\begin{enumerate}[(i)]
	\item
    first, we show that our solution to the IS-based synthesis problem, i.e., Problem~2, is sound and complete;
    \item
    then we show that   restricting Problem~1  to Problem~2 is without loss of generality, i.e., Problem~1 is solvable if and only if Problem~2 is solvable.
\end{enumerate}
Throughout this section, we denote by $S_{\Theta^*}$ the IS-based supervisor synthesized by Algorithm~2.

\subsection{Correctness of the IS-Based Synthesis Algorithm}

In this subsection, we show that Algorithm~2 indeed solves Problem~\ref{prob2}.
First, we show that Algorithm~2 is sound in the sense that the synthesized supervisor $S_{\Theta^*}$ is opacity-enforcing.

\begin{mythm}\label{thm2}
IS-based non-deterministic supervisor $S_{\Theta^*}:(\Gamma\Sigma_o)^*\to 2^\Gamma$  encoded from $\Theta^*$ enforces opacity.
\end{mythm}
\begin{proof}
Let $s=\sigma_1\cdots\sigma_n\in P(\mathcal{L}(S_{\Theta^*}/G))$ be any observable string in    closed-loop system $S_{\Theta^*}/G$.
Let $\mathbf{m}^+_n$ be state  induced by $s$ and IS-mapping ${\Theta^*}$ according to Equation~\eqref{eq:induce}.
By Corollary~\ref{corollary1}, we know that $\bigcup M(\textbf{m}^+_n)=X_I(s)$.
According to  Algorithm~2, $\mathbf{m}^+_n$ is a reachable $Z$-state in $T^*$ by construction.
Furthermore,  since $T^*$ is obtained from $T_0$ where all $Z$-states in $Q_{reveal}$ are removed.
Therefore, we conclude that
\[
X_I(s)=\bigcup M(\textbf{m}^+)\not\subseteq X_S,
\]
which means that $\Theta^*$ enforces opacity.
\end{proof}

Note that Algorithim~2 returns ``no solution" when $y_0$ is not included in $T^*$.
Next, we show that Algorithim~2 is also complete in the sense that there is indeed no solution to Problem~2 when $y_0$ is removed by operator $F$ during the construction of $T^*$.

\begin{mythm}\label{thm3}
If there exists a non-deterministic IS-based supervisor that  enforces opacity, then  $y_0$ must be included in  $T^*$, i.e., Algorithm~2 will not return ``no solution" when a solution to Problem~2 exists.
\end{mythm}
\begin{proof}
Suppose that there exists a reachability-closed IS-mapping $\Theta: I\to 2^\Gamma$ such that the encoded non-deterministic supervisor $S_{\Theta}$ enforces opacity.
We construct a consistent G-BTS $T=(Q_Y,Q_Z,h_{Y\!Z},h_{ZY},\Sigma_o,D,y_0)$ as follows:
$Q_Y=\{\mathbf{m}:(m,\mathbf{m})\in\mathcal{I}_{S_\Theta}\}$ and for any $y\in Q_Y$,
$d= \{ (m, \Theta(m,y) ) : m\in y    \} $ is the unique macro-control-decision defined at $y$ and $Q_Z=\{ \odot(d) :  \exists y\in Q_Y\text{ s.t. } h_{YZ}(y,d)! \}$.
Since $S_{\Theta}$ enforces opacity, we have
$\forall s\in P(\mathcal{L}(S_\Theta/G)): X_I(s)\not\subseteq X_S$.
Let $\mathbf{m}^+_n$ be state  induced by $s$ and IS-mapping ${\Theta}$ according to Equation~\eqref{eq:induce}.
By the construction of $T$, we have $\mathbf{m}^+_n \in Q_Z  $.
By Corollary~\ref{corollary1}, we know that $\bigcup M(\textbf{m}^+_n)=X_I(s)$.
Therefore, we have $Q_Z \cap  Q_{reveal}=\emptyset$.
Since $T$ is included in $T_0$ and $T$ itself is consistent,
no states in $T$ can be removed when iteratively applying operator $F$, which means that all states in $T$ are also included in $T^*$.
Therefore, the initial $Y$-state $y_0$ is included in $T^*$ and Algorithm~2 will not return ``no solution".
\end{proof}

\subsection{From Non-IS-Based Supervisors to IS-Based Supervisors}
So far, we have shown that Algorithm~2 correctly solves Problem~2, which is a restrictive version of Problem~1.
Clearly, Algorithm~2 is also sound for Problem~1 because an IS-based solution is also a solution to Problem~1.
Then it remains to show the completeness of Algorithm~2 in terms of Problem~1.
To this end, it suffices to show that Problem~2 always has a solution when Problem~1 has one.
Here, we provide a constructive procedure that always construct an IS-based opacity-enforcing supervisor that solves Problem~2  when a non-IS-based one that solves Problem~1 exists.

Suppose that there exists a (possibly non-IS-based) non-deterministic supervisor $S_N: (\Gamma \Sigma_o)^* \to 2^\Gamma$ that enforces opacity.
We construct an IS-mapping $\Theta$ according to Algorithm~3.
The idea is similar to the information-flow analysis for IS-mapping, which expends the information-state space from the initial information state.
We still use $y$ to denote state-estimates immediately after an observable event and use $z$ to denote  state-estimates with the unobserable tail included.
However, since the supervisor needs not to be IS-based, simply remembering the current information-state is not sufficient and we also need to remember the history leading to each state estimate. Therefore, for each micro-state $m_i$ in a  $Y$-state, we add an additional information $\rho_i$ to track how this micro-state is visited.
Note that for each micro-state $m$ in a $Y$-state, the decision history may  not be   unique since there there may have multiple $\rho$ associated with the same $m$.
Similarly, each augmented micro-state in a $Z$-state is also attached with a decision history information.
Then procedure DoDSF implements a depth-first search to generate the domain of the IS-mapping.
Note that, since $S_N$ is not IS-based in general, it may take different actions for different histories visiting the same information-state.
Our approach is to fix the control decision for each information-state as the union of the decisions for all its visits; the constructed mapping is, therefore, forced to be IS-based.
Algorithm~3 clearly terminates in a finite number of states since it will stop when all possible macro-states are visited.

The following result shows that Algorithm~3   indeed converts a non-IS-based opacity-enforcing supervisor into an IS-based opacity-enforcing supervisor.

\begin{algorithm}[htbp]
	\SetKwData{Left}{left}
	\SetKwData{Up}{up}
	\SetKwFunction{FindCompress}{FindCompress}
	\SetKwInOut{Input}{input}
	\SetKwInOut{Output}{output}
	\Input{$S_N$}
	\Output{$\Theta$}
	\Indp
	\BlankLine
	\nl $\rho \gets\epsilon,m \gets \{x_0\} , \textsc{Visited}\gets \{  \{m\}  \}$ \;
	\nl $y\gets \{ (m, \epsilon ) \}$ \;
	\nl {DoDFS}($ y , \textsc{Visited}$)\;
	\nl
	\textbf{return} {$\Theta$}\;
	\BlankLine
	\Indm
	{\bf procedure} DoDFS($ y, \textsc{Visited} $)\;
	\Indp
	\nl
	$z\gets\emptyset$\;
	\nl
	Suppose $y=\{ (m_1, \rho_1 ),\dots, (m_k ,\rho_k ) \}$ \;
	\nl
	$ \mathbf{m}\gets \{m_1,\dots,m_k  \}  $ \;
	\nl
	\For{$i=1,\dots,k $}
	{
		\nl
		$\textsc{Dom}(\Theta)\gets\textsc{Dom}(\Theta) \cup \{(m_i , \mathbf{m}   )\}  $\;
		\nl
		$\Theta(m_i , \mathbf{m}  )\gets \bigcup\{ S_N(\rho):(m_i,\rho)\in Y\}    $ \;
		\nl
		\For{$\gamma\in S_N(\rho_i)$}
		{
			\nl
			$   z\gets z\cup \{   (U\!R_{\gamma}( m_i), \rho_i  \gamma  ,\gamma  )    \}  $\;
		}
	}
	\nl
	Suppose $z=\{ (m_1,\rho_1,\gamma_1),\dots, (m_p,\rho_p,\gamma_p) \}$ \;
	\nl
	\For{$\sigma\in\Sigma_o$}
	{
		\nl
		$y'\gets\emptyset$\;
		\nl
		\For{$i=1,\dots,p$}
		{
		\nl \If{$\sigma\in\gamma_i$}{
			\nl
			$ y'\gets   y'   \cup  \{  (N\!X_{\sigma}( m_i ),  \rho_i \sigma  )  \}  $\;}
		}
		\nl \If{$\{m:  (m,\rho)\in y'  \} \not\in \textsc{Visited}$}
		{
			\nl
			\textsc{Visited}$\gets$ \textsc{Visited}$\cup\{      \{m:  (m,\rho)\in y'  \}     \}   $\;
			\nl
			DoDFS($y'$)\;
		}
	}
	\BlankLine
	\Indm
	\caption{Construction of IS-Mapping $\Theta$ from $S_N$\label{constructM}}
\end{algorithm}

\begin{mythm}\label{thm4}
Let $S_N: (\Gamma\Sigma_o)^*\to 2^\Gamma$ be a non-deterministic supervisor enforcing opacity and $\Theta: I\to 2^\Gamma$ be the partial IS-mapping constructed by Algorithm~3. Then IS-based supervisor $S_\Theta$ also enforces opacity.
\end{mythm}
\begin{proof}
First, by construction,
for each macro-state ($Y$-state without the extended strings components) visited by IS-mapping $\Theta$, i.e., $\mathbf{m}=\{m_1,\ldots, m_k\}\in \textsc{Visited}$, IS-mapping $\Theta$ defines a non-deterministic control decision for each micro-state $m_i\in\mathbf{m}$ (lines 9-10).
Also, for each $Z$-state reached by $S_\Theta$, every possible observable events are defined (line 14).
Therefore, for every information-state $(m,\textbf{m})\in\textsc{Reach}_\Theta((\{x_0\},\{\{x_0\}\}))$, $\Theta(m,\textbf{m})$ is always well-defined, i.e.,  IS-mapping $\Theta$ is reachability-closed.
Therefore, its decoded IS-based supervisor $S_\Theta$ is well-defined and we have
\[
\mathcal{I}_{S_\Theta}=\textsc{Reach}_{\Theta}(\imath_0) = \textsc{Dom}(\Theta),
\]
where  $\imath_0=( \{x_0\},\{ \{x_0\}  \} )$.

By Corollary \ref{corollary1} and Proposition \ref{pro1}, 	
to show that  $S_\Theta$  enforces opacity, it suffices to show that
\[
\forall (m,\mathbf{m})\in \textsc{Dom}(\Theta):  \bigcup M(  \odot( d_{\Theta}( \mathbf{m} ) )  ) \not\subseteq X_S.
\]
To this end, we consider how $\mathbf{m}$ is added.
Suppose $y_0y_1\dots y_n$ is the sequence of $Y$-states  in the depth-first search
such that $y_n$ contributes $\mathbf{m}$ to $\textsc{Visited}$,  and let $s=\sigma_1\dots\sigma_n$ be the observable events along this sequence.
More clearly, suppose that $y_n=\{ (m_1, \rho_1 ),\dots,  (m_k , \rho_k ) \}$ and we have
$\{m_1,\dots,m_k\}=\mathbf{m}$.
We claim that for $y_n$, we have
\[
y_n= \{   (\hat{\mathcal{E}}_S(\rho),\rho):   \rho\in  \mathcal{O}( \mathcal{L}_e^o(S_N/G)  ) ,  \rho|_{\Sigma}=s    \}
\]
This claim can be seen inductively by the length of $s$.
For $|s|=0$, we have
$y_0=\{ (\{x_0\}, \epsilon ) \}$, where
 $\epsilon $ is the unique string in $\mathcal{O}( \mathcal{L}_e^o(S_N/G)  ) $ whose projection is also $\epsilon$
and $\hat{\mathcal{E}}_S(\rho)=\{x_0\}$.
Assume that this claim holds for $|s|=k$, then  for the case of $s\sigma_{k+1}$, according to lines 11-12  and  lines 16-17, we have
\begin{align}
y_{k+1}
=&
\left\{
\begin{array}{c c}
( N\!X_{\sigma_{k+1}}( U\!R_\gamma(  m  )    )  , \rho\gamma\sigma_{k+1}    )      : \\
  (m,\rho )\in y_k,  \gamma\in S_N(\rho), \sigma_{k+1}\in \gamma
\end{array}
\right\}   \nonumber\\
=&
\left\{
\begin{array}{c c}
(\hat{\mathcal{E}}_S(\rho'),\rho'): \\
 \rho'\in  \mathcal{O}( \mathcal{L}_e^o(S_N/G)  ) ,  \rho'|_{\Sigma}=\sigma_1\dots\sigma_k\sigma_{k+1}
\end{array}
\right\}
\nonumber
\end{align}
Now, still for the same $\mathbf{m}$ and string $s$ leading to it. We have
\begin{align}
 &M( \odot( d_{\Theta}( \mathbf{m} ) ) )   \nonumber \\
=&\{  {U\!R}_\gamma(m)\in 2^X :    (m,\Gamma)\in d_{\Theta}( \mathbf{m} ),\gamma\in\Gamma  \} \nonumber \\
=& \{  {U\!R}_\gamma( \hat{\mathcal{E}}_S(\rho)  ) \in 2^X : \rho\!\in \! \mathcal{O}( \mathcal{L}_e^o(S_N/G)  ) ,  \rho|_{\Sigma}\!=\!s,\gamma\!\in\! S_N(\rho)     \}  \nonumber \\
=&\{ \mathcal{E}_S( \rho )\in 2^X :   \rho\in \mathcal{O}(\mathcal{L}_e^d(S_N/G)) \text{ s.t. }   \rho|_{\Sigma}=s   \} \nonumber \\
=& \mathcal{E}_I(s) \nonumber
\end{align}
Since $S_N$ enforces opacity, we have $\bigcup\mathcal{E}_I(s)=X_I(s) \not\subseteq X_S$, which means that
$\bigcup M(  \odot( d_{\Theta}( \mathbf{m} ))  )\not\subseteq X_S$. This completes the proof.
\end{proof}

By combining Theorems~\ref{thm2}-\ref{thm4}, we have the following result immediately that finally establishes the correctness of the synthesis procedure.
\begin{mycol}
Algorithm~2 also correctly solves Problem~1, i.e., it is both sound and complete.
\end{mycol}

\section{Conclusion}\label{section7}

In this paper, we proposed to use non-deterministic control mechanism to enforce opacity. The essence is to leverage the non-deterministic mechanism to enhance the plausible deniability of the system.
To this end, we formally defined the non-deterministic supervisor and formulated the corresponding opacity enforcement problem. Effective approach was provided to synthesize a non-deterministic opacity-enforcing supervisor based on both the information of the supervisor and the information of the intruder. We showed that the proposed algorithm is both sound and complete in the sense that it will correctly return a non-deterministic opacity-enforcing supervisor when one exists.

Although we show that non-deterministic supervisors are strictly more powerful than  deterministic ones, the synthesis complexity is doubly-exponential in the size of the plant,
which is higher than the single-exponential complexity for the deterministic case.
Intuitively, this complexity is paid because we should not only estimate all possible states of the system from the supervisor's point of view,
but also need to estimate the supervisor's estimates from the intruder's point of view.
Recently, some new efficient  approaches, such as abstraction-based approach \cite{noori2018incremental,zhang2019opacity,hou2019abstraction,liu2020notion,yin2020approximate,liu2020verification} and compositional approach \cite{noori2018compositional,mohajerani2019transforming,mohajerani2020compositional},
have been proposed to reduce the  computational complexity in the verification and synthesis of opacity.
In the future work, we also would like to leverage these techniques to further mitigate the complexity of the proposed synthesis algorithm.

\bibliographystyle{plain}
\bibliography{des}

\end{document}